\documentclass[lettersize,journal,twoside]{IEEEtran}
\IEEEoverridecommandlockouts
\usepackage{balance}
\usepackage{cite}
\usepackage{amsmath,amssymb,amsfonts}

\usepackage[top=0.71in,bottom=1in,right=0.625in,left=0.625in]{geometry}
\usepackage{array}
\usepackage[caption=false,font=footnotesize,position=bottom]{subfig}
\usepackage{stfloats}
\usepackage{url}
\usepackage{graphicx}
\usepackage{textcomp}
\usepackage{xcolor}
\usepackage{tikz}
\usepackage[capitalise]{cleveref}
\usepackage{ellipsis}
\usetikzlibrary{calc}
\usetikzlibrary{decorations.pathreplacing,decorations.markings,shapes.geometric}
\usetikzlibrary{calc,patterns,angles,quotes,shapes,arrows.meta}
\usetikzlibrary{chains,spy}
\usepackage{multirow}
\usepackage{amsthm}
\newtheorem{theorem}{Theorem}

\newtheorem{corollary}{Corollary}[theorem]

\usepackage{algorithm}
\usepackage{algpseudocode}

\usepackage[utf8]{inputenc}

\usepackage{mathtools}
\usepackage{bm}
\usepackage{etoolbox}
\usepackage{scalerel}

\usepackage{fancyhdr}

\usepackage{makecell}

\usepackage[acronym,shortcuts]{glossaries}
\newacronym{GMM}{GMM}{Gaussian mixture model}
\newacronym{PDF}{PDF}{probability density function}
\newacronym{MSE}{MSE}{mean square error}
\newacronym{CSI}{CSI}{channel state information}
\newacronym{CME}{CME}{conditional mean estimator}
\newacronym{ML}{ML}{maximum likelihood}
\newacronym{LS}{LS}{least squares}
\newacronym{LOS}{LoS}{line-of-sight}
\newacronym{NLOS}{NLoS}{non-\ac{LOS}}
\newacronym{DoA}{DoA}{direction-of-arrival}
\newacronym{SNR}{SNR}{signal-to-noise ratio}
\newacronym{BS}{BS}{base station}
\newacronym{JCAS}{JCAS}{joint communication and sensing}
\newacronym{MIMO}{MIMO}{multiple-input-multiple-output}
\newacronym{SIMO}{SIMO}{single-input-multiple-output}
\newacronym{MMSE}{MMSE}{minimum \ac{MSE}}
\newacronym{NMSE}{NMSE}{normalized \ac{MSE}}
\newacronym{RMSE}{RMSE}{root \ac{MSE}}
\newacronym{LMMSE}{LMMSE}{linear \ac{MMSE}}
\newacronym{MT}{MT}{mobile terminal}
\newacronym{UE}{UE}{user equipment}
\newacronym{OMP}{OMP}{orthogonal matching pursuit}
\newacronym{CS}{CS}{compressed sensing}
\newacronym{ULA}{ULA}{uniform linear array}
\newacronym{URA}{URA}{uniform rectangular array}
\newacronym{DFT}{DFT}{discrete Fourier transform}
\newacronym{MUSIC}{MUSIC}{multiple signal classification}
\newacronym{ESPRIT}{ESPRIT}{estimation of signal parameters via rotational invariance techniques}
\newacronym{GE}{GE}{gridded estimator}
\newacronym{AWGN}{AWGN}{additive white Gaussian noise}
\newacronym{GSC}{GSC}{generalized sidelobe cancellor}
\newacronym{EM}{EM}{expectation maximization}
\newacronym{MP}{MP}{message passing}
\newacronym{VAE}{VAE}{variational autoencoder}
\newacronym{MFA}{MFA}{mixtures of factor analyzers}
\newacronym{ELBO}{ELBO}{evidence-lower bound}
\newacronym{KL}{KL}{Kullback-Leibler}
\newacronym{DNN}{DNN}{deep neural network}
\newacronym{CGLM}{CGLM}{conditionally Gaussian latent model}
\newacronym{CRB}{CRB}{Cramer-Rao bound}
\newacronym{BCRB}{BCRB}{Bayesian Cramer-Rao bound}
\newacronym{TDS}{TDS}{time domain synchronous}
\newacronym{OFDM}{OFDM}{orthogonal frequency division multiplexing}
\newacronym{OTFS}{OTFS}{orthogonal time frequency space}
\newacronym{MAP}{MAP}{maximum a-posteriori}
\newacronym{PSD}{PSD}{positive semidefinite}
\newacronym{wlog}{w.l.o.g.}{without loss of generality}
\newacronym{PAST}{PAST}{projection approximation subspace tracking}
\newacronym{wrt}{w.r.t.}{with respect to}
\newacronym{SVD}{SVD}{singular value decomposition}

\newcommand{\va}{{\bm{a}}}

\newcommand{\vc}{{\bm{c}}}
\newcommand{\vd}{{\bm{d}}}

\newcommand{\vh}{{\bm{h}}}

\newcommand{\vn}{{\bm{n}}}

\newcommand{\vs}{{\bm{s}}}

\newcommand{\vv}{{\bm{v}}}

\newcommand{\vx}{{\bm{x}}}
\newcommand{\vy}{{\bm{y}}}
\newcommand{\vz}{{\bm{z}}}

\newcommand{\ma}{{\bm{A}}}
\newcommand{\mb}{{\bm{B}}}
\newcommand{\mc}{{\bm{C}}}
\newcommand{\md}{{\bm{D}}}

\newcommand{\mh}{{\bm{H}}}

\newcommand{\mn}{{\bm{N}}}

\newcommand{\matp}{{\bm{P}}}
\newcommand{\mq}{{\bm{Q}}}
\newcommand{\mr}{{\bm{R}}}
\newcommand{\ms}{{\bm{S}}}

\newcommand{\mv}{{\bm{V}}}
\newcommand{\mw}{{\bm{W}}}
\newcommand{\mx}{{\bm{X}}}
\newcommand{\my}{{\bm{Y}}}

\newcommand{\vth}{{\bm{\theta}}}

\newcommand{\vsig}{\bm{\sigma}}

\newcommand{\vphi}{{\bm{\phi}}}

\newcommand{\vmu}{\bm{\mu}}

\newcommand{\vdel}{{\bm{\delta}}}

\newcommand{\tr}{\mathrm{tr}}

\newcommand{\diag}{\mathrm{diag}}

\newcommand{\range}{\mathrm{range}}
\newcommand{\mse}{\mathrm{MSE}}
\newcommand{\proj}{{\mathrm{proj}}}
\newcommand{\sub}{{\mathrm{sub}}}

\newcommand{\He}{\mathrm{H}}
\newcommand{\T}{\mathrm{T}}

\newcommand{\I}{\mathbf{I}}
\renewcommand{\j}{\mathrm{j}}
\newcommand{\E}{\mathbb{E}}
\newcommand{\NC}{\mathcal{N}_\mathbb{C}}

\tikzstyle{block} = [draw, rectangle, 
minimum height=4em, minimum width=4em]
\tikzstyle{input} = [coordinate]
\tikzstyle{output} = [coordinate]
\tikzstyle{pinstyle} = [pin edge={to-,thin,black}]

\usetikzlibrary{positioning}

\tikzset{radiation/.style={{decorate,decoration={expanding waves,angle=90,segment length=5pt}}}}

\usetikzlibrary{spy}
\usepackage{pgfplots}
\usepgfplotslibrary{fillbetween}
\usepackage{wrapfig}
\usetikzlibrary{arrows,shapes}
\usetikzlibrary{positioning,shapes.callouts}
\usepgfplotslibrary{groupplots,dateplot}
\usetikzlibrary{patterns,shapes.arrows}
\pgfplotsset{compat=newest}

\tikzset{PlotML/.style={mark=star,mark size=2.2pt, line width=1pt, color=gray, dashed, mark options=solid}}
\tikzset{PlotProjsCov/.style={mark=triangle,mark size=1.5pt, line width=1pt, color=green!40!black, dashed, mark options={solid, rotate=180}}}
\tikzset{PlotSubsCov/.style={mark=diamond,mark size=1.8pt, line width=1pt, dashed, color=green!50!teal, mark options=solid}}
\tikzset{PlotGMM/.style={mark=o, mark size=1.5pt,domain=1:10000, line width=1pt, color=orange, dashed, mark options=solid}}
\tikzset{PlotGMMHard/.style={mark=diamond, mark size=1.8pt,domain=1:10000, line width=1pt, color=purple, dashed, mark options=solid}}
\tikzset{PlotProjGMM/.style={mark=triangle,mark size=1.5pt, line width=1pt, color=blue, mark options=solid}}
\tikzset{PlotSubGMM/.style={mark=square,mark size=1.5pt, line width=1pt, color=cyan, mark options=solid}}
\tikzset{PlotEM/.style={mark=asterisk, mark size=2.2pt, line width=1pt, color=yellow!50!orange, dashed, mark options=solid}}
\tikzset{PlotMP/.style={mark=x, mark size=2.2pt, line width=1pt, color=magenta, dashed, mark options=solid}}
\tikzset{PlotProjVAE/.style={mark=triangle,mark size=1.5pt, line width=1pt, color=violet, mark options={solid, rotate=270}}}
\tikzset{PlotSubVAE/.style={mark=pentagon,mark size=1.8pt, line width=1pt, color=red, mark options={solid, rotate=0}}}
\tikzset{PlotLS/.style={mark=|,mark size=2pt,line width=1pt, color=black, dashed, mark options=solid}}
\tikzset{PlotVAE/.style={mark=triangle,mark size=1.5pt, line width=1pt, color=olive, dashed, mark options={solid, rotate=90}}}

\def\BibTeX{{\rm B\kern-.05em{\sc i\kern-.025em b}\kern-.08em
		T\kern-.1667em\lower.7ex\hbox{E}\kern-.125emX}}

\newcommand{\update}[1]{\textcolor{black}{#1}}
    
\usepackage{fancyhdr}

\fancypagestyle{cfooter}{ %
	\fancyhf{} 
	\cfoot{	\small This work has been submitted to the IEEE for possible publication. Copyright may be transferred without notice, after which this version may no longer be accessible.
	}
	
}

\begin{document}

\title{
Semi-Blind Strategies for MMSE Channel Estimation Utilizing Generative Priors
}
\author{\IEEEauthorblockN{Franz Weißer,~\IEEEmembership{Graduate Student Member,~IEEE,} Nurettin Turan,~\IEEEmembership{Graduate Student Member,~IEEE,} \\Dominik Semmler,~\IEEEmembership{Graduate Student Member,~IEEE,} Fares Ben Jazia, and Wolfgang Utschick,~\IEEEmembership{Fellow,~IEEE}%
\thanks{This work was supported by the Federal Ministry of Education and Research of Germany in the programme of “Souverän. Digital. Vernetzt.”. Joint project 6G-life, project identification number: 16KISK002.
An earlier version of this work was presented at ICASSP’24\cite{Weisser2024}.
}
\thanks{The authors are with the TUM School of Computation, Information and Technology, Technical University of Munich, 80333 Munich, Germany (e-mail: franz.weisser@tum.de).}}
}

\maketitle

\thispagestyle{cfooter}

\begin{abstract}
This paper investigates semi-blind channel estimation for massive multiple-input multiple-output (MIMO) systems.
To this end, we first estimate a subspace based on all received symbols (pilot and payload) to provide additional information for subsequent channel estimation.
This additional information enhances minimum mean square error (MMSE) channel estimation.
Two variants of the linear MMSE (LMMSE) estimator are formulated, where the first one solves the estimation within the subspace, and the second one uses a subspace projection as a preprocessing step.
Theoretical derivations show the latter method's superior estimation performance in terms of mean square error for uncorrelated Rayleigh fading.
\update{Further, we provide asymptotical insights on how the proposed MMSE-based channel estimation strategy outperforms the unbiased Cramer-Rao bound.}
Subsequently, we introduce parameterizations of these semi-blind LMMSE estimators based on two different conditional Gaussian latent models, i.e., the Gaussian mixture model and the variational autoencoder.
Both models learn the propagation environment's underlying channel distribution based on training data 
and serve as generative priors for our semi-blind channel estimation.
Extensive simulations for real-world measurement data and spatial channel models show the proposed methods' superior performance compared to state-of-the-art semi-blind channel estimators in terms of MSE.
\end{abstract}


\section{Introduction}

\IEEEPARstart{A}{ccurate} \ac{CSI} is crucial for achieving the expected high data rates promised by \ac{MIMO} systems~\cite{Rusek2013,Kabalci2019,Ye2018}. 
The \ac{CSI} describes the communication link between transmitter and receiver, characterized by its time-varying and frequency-selective nature, which is prone to rapid changes, making the task of channel estimation complex~\cite{Shang2021}.
As accurate channel estimates are essential for successfully transmitting data, it is at the center of several research efforts~\cite{Harkat2022,Smera2021}.

The most widely adopted methods in wireless communications utilize known training or pilot symbols
transmitted across the channel using some of the radio resource blocks~\cite{Marzetta2006}. 
Afterward, the receiver uses the observed signals to determine a reliable \ac{CSI} estimate. 
As the number of pilots scales with the number of users, the spectral efficiency decreases for a higher number of users as fewer symbols are available for transmitting data.
To enhance channel estimation without increasing the number of pilot symbols, various methods have been developed that leverage the information embedded in the observed data symbols at the receiver to infer channel characteristics~\cite{DeCarvalho1997a,DeCarvalho1997b,Medles2003,Ma2014,Joham1999,Neumann2015,Nayebi2018,Liu2009,Wu2016,Mehrotra2023,Osinsky2020,Liu2012,Park2015,Khan2023, Kim2023,Khan2024,Zilberstein2024}.
These methods exploit structure and redundancy within the transmitted data and yield more accurate \ac{CSI} estimates.

The benefit of semi-blind channel estimation was first studied in~\cite{DeCarvalho1997a}, where \acp{CRB} for blind, semi-blind, and training-based channel estimation were investigated in the context of \ac{SIMO} systems.
In~\cite{DeCarvalho1997b, Medles2003}, semi-blind channel estimation schemes were introduced based on \ac{ML} estimation. 
The respective estimators' asymptotic performances were also studied in~\cite{DeCarvalho1997b} for infinitely long data sequences.
Another asymptotic behavior, where the number of antennas grows to infinity, was studied in~\cite{Ma2014}. Here, the authors identified two interference components in semi-blind channel estimation, which do not vanish even for large numbers of antennas.
Early work on improving the \ac{LS} estimator using a semi-blind algorithm was conducted in~\cite{Joham1999}.
Furthermore, in~\cite{Neumann2015}, semi-blind and blind channel estimation was studied to enhance the \ac{MAP} channel estimates in massive \ac{MIMO} systems. 
These findings are based on favorable propagation, which only holds for large antenna arrays deployed at the \ac{BS}.
In~\cite{Nayebi2018}, two semi-blind channel estimators are studied based on the \ac{EM} algorithm. 
The assumption of a Gaussian distribution for the data was verified, leading to a closed-form solution for the E-step.
Another iterative framework optimizing the likelihood based on \ac{MP} is used in~\cite{Liu2009,Wu2016} and references therein.
In~\cite{Mehrotra2023}, a data-aided iterative scheme is proposed for \ac{OTFS} systems by employing affine-precoded superimposed pilots.
The performance improvement achieved with these iterative approaches generally requires computationally costly updates. 
A low complexity iterative \ac{LS} channel estimation algorithm is proposed in~\cite{Osinsky2020} for a massive \ac{MIMO} turbo-receiver.
In~\cite{Liu2012}, the concept of semi-blind channel estimation was adapted for time domain synchronous-\ac{OFDM} systems, where in addition to the pseudo noise sequence in the guard interval, the sent \ac{OFDM} data symbols are exploited for the channel estimation.
\update{In~\cite{Park2015,Khan2023,Kim2023}, different decision-directed frameworks were proposed, which treat reliably decoded data symbols as additional pilots.
All of these methods utilize iterative approaches, thus exhibiting large complexities.}
In~\cite{Khan2024}, peak-power carriers in an \ac{OFDM} system are selected to eliminate the need to determine reliable data symbols at the receiver.
Recently, a diffusion model-based approach for joint channel estimation and detection was proposed in~\cite{Zilberstein2024}, where a diffusion process is constructed that models the joint distribution of the channels and symbols given noisy observations.

\update{Many state-of-the-art estimators fall in the class of unbiased estimators, which are potentially \ac{MSE}-suboptimal.
Instead,} in this article, we focus on \update{general, potentially biased,} pilot-based estimators that minimize the \ac{MSE} \update{by utilizing prior information about the channel distribution,} and investigate how these estimators can be extended for the semi-blind case.
\update{Contrarily to previous works, e.g.,~\cite{Park2015,Khan2023,Kim2023}, we do not focus on using decoded data symbols as additional pilot but instead want to improve the estimators themselves.
Although both approaches can be combined.}
Notably, in~\cite{Deng2020} a subspace formulation of the \ac{MMSE} estimator was already used to mitigate pilot contamination in massive \ac{MIMO} systems.
The \ac{MMSE} estimator is known to be the \ac{CME}~\cite[Ch. 10]{Kay1993}, which, in general, is intractable and can not be computed in closed form.
Recently, powerful approximations based on machine learning were presented in~\cite{Yang2015,Neumann2018,Koller2022, Fesl2023c,Baur2024}.
The benefit of machine learning is to enhance the task at hand by using prior information captured during the learning stage.
For a given \ac{BS} cell environment, the \ac{PDF} representing potential user channels can be considered valuable prior information.
Since this true underlying distribution is unknown, machine learning methods rely on a representative data set, which is assumed to be available at the \ac{BS}.
Based on this data set, the user channels' \ac{PDF} can be learned.
The first proposal of using a \ac{GMM} to formulate an estimator was done in~\cite{Yang2015} for the case of image processing.
The approaches in~\cite{Koller2022, Fesl2023c,Baur2024} build on that by constructing a \ac{CGLM} for the \ac{PDF} of a BS cell environment.
The learned \ac{CGLM} not only enables \ac{MMSE} channel estimation in~\cite{Koller2022, Fesl2023c, Baur2024} but can also be used for, e.g., a limited feedback scheme as in~\cite{Turan2023}.
In this work, we propose to utilize \acp{CGLM} to parameterize the \ac{CME} in the semi-blind setting.

This work's contributions are summarized as follows:
\begin{itemize}
    \item 
    We introduce two variants of the \ac{LMMSE} estimator incorporating subspace knowledge provided by the payload data symbols.
    Firstly, we depict how the \ac{LMMSE} channel estimator can solve a subspace estimation problem~\cite{Deng2020}. 
    As an alternative, we propose a new projection method that is computationally more efficient since it allows for the pre-calculation of \ac{LMMSE} filters.
    \item 
    With theoretical derivations, we show the proposed projection method's superior \ac{MSE} performance in the case of uncorrelated Rayleigh fading and perfect subspace knowledge.
    \item \update{We provide an asymptotical analysis of how the projected \ac{LMMSE} outperforms the unbiased \ac{CRB} for perfect subspaces.
    These derivations highlight the \ac{MMSE}-based channel estimation's superiority compared to unbiased channel estimation in the case of semi-blind channel estimation.}
    \item 
    We show how the \ac{GMM}~\cite{Koller2022} and \ac{VAE}~\cite{Baur2024}, instances of the class of \acp{CGLM}, can be used to parameterize the semi-blind \ac{LMMSE} estimator.
    \item 
    Extensive simulations on different datasets, consisting of typical massive \ac{MIMO} systems with multiple users and real-world measurement data, show our proposed methods' superior performance compared to state-of-the-art semi-blind channel estimation algorithms in terms of the \ac{MSE}.
\end{itemize}

Preliminary results were presented in~\cite{Weisser2024} and extended to the multi-user \ac{MIMO} case in~\cite{Weisser2024a}, which we extend further in the following aspects.
The theoretical analyses in \cref{sec:perf} enhance the proposed semi-blind channel estimation strategies' foundation and provide analytic insights into the proposed projection method's superior performance.
We extend our concept of semi-blind \ac{MMSE} channel estimation to the whole class of \acp{CGLM}, providing a more general framework to parameterize the semi-blind \ac{LMMSE} estimator.
Finally, we provide more comprehensive simulation results to show our proposed strategies' strengths.

\emph{Notations:} Matrices and vectors are denoted with boldface symbols. $\bm{0}$ and $\I_N$ denote the zero vector of appropriate size and the identity matrix of size $N\times N$, respectively. $\E[\cdot]$, $\tr(\cdot)$, and $\range(\cdot)$ denote the expectation, trace, and range operators, respectively. We use $(\cdot)^\T$, $(\cdot)^\He$, $(\cdot)^{-1}$, \update{$(\cdot)^{\dagger}$} to denote the transpose, conjugate transpose, inverse, \update{and pseudo inverse}. $\|\cdot\|$ denotes the $2$-norm of a vector. $\NC(\vmu,\mc)$ denotes the circularly symmetric complex Gaussian distribution with mean $\vmu$ and covariance matrix $\mc$.

\section{System and Channel Model}
\label{sec:sys}
\vspace{-1pt}

\update{
We consider a multi-user system with $J$ single-antenna users and a \ac{BS} equipped with $M$ receive antennas.
We focus on the uplink scenario for this work, but the proposed methods can also be extended to the downlink.} 
The received signal vector at time instance $n$ is then
\begin{align}
	\vy(n) 
	&=  \mh \vx(n) + \vn(n) , \quad n=1,...,N, \label{eq:sys_model}
\end{align}
where $\vx(n) = [x_1(n),...,x_J(n)]^\T \in \mathbb{C}^J$ and $\vn(n) \in \mathbb{C}^M$ denote the signal sent by each of the $J$ users and the additive noise, respectively, whereas $\mh = [\vh_1, ..., \vh_J]$ contains the individual user channels $\vh_j \in \mathbb{C}^M$.
The case of multiple antennas at the users can be transformed into~\eqref{eq:sys_model} by viewing each active stream as a different user with its corresponding channel $\vh_j$ being the respective effective channel.
For further details on semi-blind channel estimation in multi-user \ac{MIMO}, we refer the reader to~\cite{Weisser2024a}.
\update{For the case of downlink semi-blind channel estimation, a subspace of 
the effective channel needs to be identifiable to improve estimation performance. This is the case if not all degrees of freedom are utilized, e.g., fewer pilots than receive antennas.}
We consider a channel coherence interval larger than the number of snapshots $N$, i.e., the channels are constant over all snapshots. 
We assume that the noise is Gaussian with $\vn(n) \sim \NC (\bm{0}, \mc_\vn = \sigma^2\I_M)$ \update{and the sent signals satisfy $\mathbb{E}[\vx(n)\vx^\He(n)]=\frac{1}{J}\I_J$}.

In conventional channel estimation schemes, each user's signals include $N_p$ uplink pilots. These pilots are known to the \ac{BS}. 
Hence, the received observations at the \ac{BS} side are 
\begin{align}
 \my= \left[\my_p^\prime, \my_d\right] = \mh\left[\matp, \md\right] +\mn= \mh\mx +\mn,
	\label{Eq:AllObservations}
\end{align}
where $\my\in\mathbb{C}^{M\times N}$, $\my^\prime_p\in\mathbb{C}^{M\times N_p}$, $\my_d\in\mathbb{C}^{M\times N-N_p}$, $\matp\in\mathbb{C}^{J\times N_p}$, and \update{$\md=[\vd_1,\dots,\vd_J]^\T$}
denote all received observations, received pilot observations, received payload data observations, sent pilots, and sent payload data symbols, respectively.
\update{The $j$-th user's payload symbols $\vd_j\in\mathbb{S}_j^{N-N_p}$ include realizations of the used transmit constellation, where $\mathbb{S}_j$ denotes the set of constellation points of user $j$.
In general, the employed user constellations can differ from each other.}
In order to fully illuminate the channels, the number of pilots is, at minimum, the number of users $N_p\geq J$, and orthogonal pilots are used.
We set \update{$\matp\matp^\He=\frac{N_p}{J}\I_M$}, and utilize \update{submatrices of the \ac{DFT} matrix of size $N_p\times N_p$}.
After decorrelating the orthogonal pilot sequences, the received pilot observations simplify to
\begin{align}
    \my_p= \my^\prime_p\matp^{\update{\dagger}}  = \mh \matp \matp^{\update{\dagger}} + \mn \matp^{\update{\dagger}} = \mh + \update{\sqrt{\frac{J}{N_p}}}\mn_p, \label{eq:pilotsys}
\end{align}
where $\mn_p$ has the same statistics as $\mn$ and, hence, we can omit the subscript.
This decorrelation lets us consider channel estimation from a per-user perspective in the subsequent discussions. 
For reasons of simpler readability, the index for the respective user is, therefore, no longer given in the following.
Consequently, we denote a user's pilot observation as 
\begin{align}
	\vy_p = \vh + \vn, \label{eq:pilotsys_su}
\end{align}
with \update{$\vn \sim \NC(\bm{0},\mc_\vn= \sigma^2_{\mathrm{eff}}\I_M)$, where 
$\sigma^2_{\mathrm{eff}}=\sigma^2\frac{J}{N_p}$
denotes the effective noise variance.}
\update{Further, we assume normalized channels as $\mathbb{E}[\|\vh\|^2]=M$, which lets us define the \ac{SNR} as $\sigma^{-2}$.}
\update{As the number of pilots $N_p$ results in a simple scaling of \ac{SNR} in the pilot transmission phase, we use $N_p=J$ in the rest of this work.}

\subsection{Spatial Channel Model}
\label{sec:3gpp}

We consider a spatial channel model based on~\cite{3GPP2020}, where the channel vectors are considered as conditionally Gaussian distributed~\cite{Neumann2018} 
\begin{align}
	\vh \mid \vdel \sim \NC \left(\bm{0}, \mc_\vdel\right),
\end{align}
based on a set of parameters $\vdel$, which describe the directions and properties of the multi-path propagation clusters.
The central angles are drawn independently from a uniform distribution in $[0,2\pi]$, and the path gains are independent zero-mean Gaussians.
The spatial covariance matrix is given by
\begin{align}
	\mc_\vdel = \int_{-\pi}^{\pi} g(\vartheta,\vdel) \va(\vartheta)\va^\He(\vartheta) \mathrm{d}\vartheta,
\end{align}
where $g(\vartheta,\vdel)$ is the power density consisting of a weighted sum of Laplace densities, which have standard deviations $\sigma_\text{AS}$ corresponding to the propagation clusters' angular spread.
The \ac{BS} employs a \ac{ULA} with $M$ antennas and $\lambda/2$ spacing.
The steering vector is then given as
\begin{align}
	\va(\vartheta) = \frac{1}{\sqrt{M}}\left[1, \mathrm{e}^{-\j\pi\sin(\vartheta)}, ..., \mathrm{e}^{-\j\pi(M-1)\sin(\vartheta)}\right]^\T. \label{eq:spaceSteering}
\end{align}
We consider a new $\vdel$ for every channel sample and draw the sample according to $\vh \sim \NC \left(\bm{0}, \mc_\vdel\right)$.

\subsection{Measurement Campaign}
\label{sec:meas}

Since synthetic data captures real-world \ac{CSI} characteristics only up to some extent, we utilize real-world data from a measurement campaign conducted at the Nokia campus in Stuttgart, Germany, in October/November 2017, 
cf.~\cite{Turan2022}. 
The \ac{BS} antenna with a \ac{URA} comprises $N_v=4$ vertical ($\lambda$ spacing) and $N_h=16$ horizontal ($\lambda/2$ spacing) single polarized patch antennas.
The operating carrier frequency is $2.18$ GHz, and the antenna was mounted on a rooftop approximately $20$ meters above the ground. 
For further details, we refer the reader to~\cite{Turan2022}. 

\section{Semi-Blind Channel Estimation using Perfect Statistical Knowledge}
\label{sec:perf}

In this section, we introduce two variants of the \ac{LMMSE} estimator incorporating information provided by the payload data symbols. 
To do so, we restrict ourselves \update{for now} to the case where perfect statistical knowledge is available at the receiver.

In channel estimation, commonly, only the pilot observation $\vy_p$ is considered for channel estimation.
The \ac{MSE}-optimal estimator given the pilot observation $\vy_p$ is the \ac{CME} 
\begin{align}
    \hat{\vh}_\text{CME} = \mathbb{E}\left[\vh\mid\vy_p\right]. \label{eq:cme_short}
\end{align}
If the channel sample $\vh$ is drawn from a Gaussian distribution according to
\begin{align}
    \vh \sim \NC (\bm{0},\mc), \label{eq:gaussModel}
\end{align}
and if further this statistic is known at the receiver,
the genie-aided \ac{CME} can be formulated as~\cite[Ch. 10]{Kay1993}
\begin{align}
    \hat{\vh}_\text{CME} 
    = \mc \left(\mc + \mc_\vn\right)^{-1} \vy_p. \label{eq:cme_lmmse}
\end{align}
This \ac{LMMSE} estimator
achieves the \ac{MSE} of
\begin{align}
    \mse^\mathrm{plain} = \tr \left[\mc - \mc(\mc + \mc_\vn)^{-1}\mc\right]. \label{eq:mse}
\end{align}
For $\mc_\vn = \sigma^2\I_M$, the \ac{MSE} can be expressed using the Woodbury identity as
\begin{align}
    \mse^\mathrm{plain} = \sum_{i=1}^M \frac{{\rho}_i \sigma^2}{{\rho}_i + \sigma^2}, \label{eq:mse2}
\end{align}
where $\rho_i$ are the eigenvalues of $\mc$.
\update{In contrast to \ac{MMSE}-based channel estimation, \ac{ML}-based channel estimation does not utilize knowledge about the channels' prior distribution, making it a suboptimal estimation problem.
}

For our considerations concerning semi-blind channel estimation, we assume knowledge about the subspace defined by $\range(\mh)=\range(\mv)$, where we denote with $\mv$ the left singular vectors of $\mh$, which span the same subspace as the columns of $\mh$.
Generally, we can formulate the \update{(suboptimal)} \ac{ML} estimate of $\mh$ in view of~\eqref{Eq:AllObservations} as~\cite{Nayebi2018}
\begin{align}
    &\min_{\mh,\mx}\sum_{n=1}^J \left\|\vy(n) - \vh_n\right\|^2 
    + \sum_{n=J+1}^N \left\|\vy(n) - \mh\vx(n)\right\|^2,\label{eq:ml_opt}
\end{align}
where the first term belongs to the pilot observations, and the second part refers to the observation obtained from the payload symbols.
\update{Obtaining a closed-form solution for~\eqref{eq:ml_opt} is known to be hard~\cite{Dempster1977}.}
In~\cite{Nayebi2018}, the \ac{EM} algorithm is introduced to solve this channel estimation problem in terms of maximum likelihood, whereas, in~\cite{Medles2003}, a semi-blind method is derived based on utilizing the subspace $\range(\mh)=\range(\mv)$.
Now, given the subspace $\range(\mv)$, we can reformulate \eqref{eq:ml_opt} as~\cite{Medles2003}
\begin{align}
    \min_{\ms,\mx} \sum_{n=1}^J \left\|\vy(n) - \mv\vs_n\right\|^2 
    + \hspace{-1pt}\sum_{n=J+1}^N \hspace{-1pt}\left\|\vy(n) - \mv\ms\vx(n)\right\|^2,\label{eq:ml_opt2}
\end{align}
with $\mh=\mv\ms$ and $\ms=[\vs_1,\dots,\vs_J]\in \mathbb{C}^{J\times J}$.
\update{We can rewrite the second term in~\eqref{eq:ml_opt2} as
\begin{align}
    \sum_{n=J+1}^N \hspace{-5pt}\left\|\vy(n) - \mv\ms\vx(n)\right\|^2
    &= \hspace{-5pt}\sum_{n=J+1}^N \hspace{-5pt}\left\|\left(\I_M - \mv\mv^\He\right)\vy(n)\right\|^2 \nonumber\\
    &\quad+ \left\|\mv^\He\vy(n) - \ms\vx(n)\right\|^2,\label{eq:orth}
\end{align}
where the equality comes from the fact that the two terms are orthogonal to each other.
The first term in~\eqref{eq:orth} is constant for given $\mv$.
In the case of a continuous constellation for $\vx(n)$, the second term vanishes by solving for $\mx$.
Further, in the case of discrete symmetric signal constellations, the second term in~\eqref{eq:orth} can not be jointly solved for $\ms$ and $\mx$ due to phase ambiguity.
}%
Thus, the \ac{ML} problem for the user of interest \update{reduces to the pilot observations as}
\cite{Medles2003}
\begin{align}
    \min_\vs\|\vy_p - \mv\vs\|^2, \label{eq:ml_est}
\end{align}
with the closed form solution $\hat{\vh}_\text{ML} = \mv\mv^\He\vy_p$.
This estimator's \ac{MSE} is
\begin{align}
    \mse^\mathrm{ML} =  \E\left[\|\vh - \mv\mv^\He\vy_p\|^2\right] = J\sigma^2. \label{eq:mse_ml}
\end{align}

In the following, we introduce two channel estimation strategies combining the information provided by $\mc$ and $\range(\mv)$.

\subsection{Subspace Channel Estimator}
\label{sec:sub}

Using the information in $\range({\mv})$, we can solve the estimation within the subspace as previously proposed in~\cite{Deng2020}.
For this, the pilot system model in \eqref{eq:pilotsys_su} is mapped into the $J$-dimensional subspace as
\begin{align}
	\vy^\prime = \mv^\He\vy_{p} &= \mv^\He\vh + \mv^\He \vn = \vh^\prime + \vn^\prime. \label{eq:sub}
\end{align}
Under the assumption that $\mv$ is  chosen independently of $\vh$, the distribution $\vh^\prime\sim\NC(\bm{0},\mv^\He\mc\mv)$ can be used to formulate the estimate~\cite{Deng2020}
\begin{align}
	\hat{\vh}^\prime = \;&\mv^\He\mc\mv \left(\mv^\He\mc\mv + \sigma^2\I_J\right)^{-1}\mv^\He\vy_p. \label{eq:condsub}
\end{align}
By design, $\mv$ depends on $\vh$ and, hence,~\eqref{eq:condsub} is a suboptimal but feasible estimate for $\vh^\prime$.
After solving the estimation in the subspace for $\vh^\prime$, the solution can be transformed back using~\cite{Deng2020}
\begin{align}
	\hat{\vh}_\text{sub} = \mv \hat{\vh}^\prime. \label{eq:sub2}
\end{align}

\subsection{Projected Channel Estimator}
\label{sec:proj}

As an alternative approach, we propose using the orthogonal subspace projection $\matp_\mh=\mv\mv^\He$ as a preprocessing filter.
Since the projector $\matp_\mh$ does not affect $\vh$, the resulting projected pilot observation is given by
\begin{align}
	\tilde{\vy} = \matp_\mh\vy_p = \vh + \matp_\mh \vn = \vh + \tilde{\vn}.
\end{align}
To formulate the \ac{CME} $\hat{\vh} = \mathbb{E}\left[\vh\mid\tilde{\vy}\right]$, we calculate the \update{covariance matrix} of 
\update{$\tilde{\vy}$ as
\begin{align}
    \mathbb{E}\left[\tilde{\vy}\tilde{\vy}^\He\right] 
    &= \mc + \mathbb{E}\left[\sigma^2\matp_\mh\right], \label{eq:cov_y}
\end{align}
where the mixing terms vanish due to the independence of $\vh$ and $\vn$.
}
To get an intuitive understanding of \update{the noise covariance in}~\eqref{eq:cov_y}, let us consider a scenario involving spatially uncorrelated channels, meaning that path gains and channel directions are uncorrelated. 
This is the case when users are uniformly distributed over the directions, e.g., the spatial channel model in Section~\ref{sec:3gpp}, resulting in a scenario's channel covariance matrix that is a scaled identity~\cite[Def. 2.3]{Bjoernson2017}. 
In such a case, the matrices with the sample covariance matrix's eigenvectors are distributed with Haar measure~\cite[Chap. 1]{Milman1986}, i.e., uniformly distributed on the manifold of unitary matrices.
Assuming spatially uncorrelated channels 
we have
\begin{align}
	\mc_{\tilde{\vn}} = \sigma^2\frac{J}{M}\I_M, \label{eq:noise_tilde}
\end{align}
which we assume to hold for the remainder of this section.
We can then formulate the projected \ac{LMMSE} estimator as
\begin{align}
	\hat{\vh}_\text{proj} = \mc \left(\mc + \mc_{\tilde{\vn}}\right)^{-1}\tilde{\vy}. \label{eq:proj}
\end{align}

\subsection{Performance Analysis \update{for Perfect Subspace Knowledge}}

If~\eqref{eq:noise_tilde} is true, the \ac{MSE} of the proposed projected channel estimator can directly be written as (cf.~\cref{app:projMSE})
\begin{align}
    \mse^\proj &= \tr \left({\mc} - {\mc}\left({\mc} + \sigma^2\frac{J}{M}\I_M\right)^{-1}{\mc}\right) \label{eq:mse_proj}\\
    &= \sum_{i=1}^M \frac{\rho_i\sigma^2}{\frac{M}{J}\rho_i + \sigma^2}. \label{eq:mse_proj2}
\end{align}
Comparing the performance to the plain \ac{LMMSE} from~\eqref{eq:mse2} we see that
\begin{align}\label{eq:mse_proj3}
    \frac{\rho_i\sigma^2}{\frac{M}{J}\rho_i + \sigma^2}\leq \frac{\rho_i\sigma^2}{\rho_i + \sigma^2},
\end{align}
holds for every $i\in\{1,\dots,M\}$, resulting in $\mse^\proj \leq \mse^\mathrm{plain}$.
The inequality in \eqref{eq:mse_proj3} only holds with equality if $J=M$.
Additionally, we can compare the \ac{MSE} of the projected \ac{LMMSE} to the \ac{ML} estimator from~\eqref{eq:mse_ml} by reformulating~\eqref{eq:mse_proj2} as
\begin{align}
     \mse^\proj = 
    \frac{J}{M}\sigma^2\sum_{i=1}^M \frac{\rho_i}{\rho_i + \frac{J}{M}\sigma^2} \leq J\sigma^2 = \mse^\text{ML}.
\end{align}

In the case of uncorrelated Rayleigh fading, the channel covariance matrix is given as $\mc=\I_M$.
The projected variant's \ac{MSE} results in
\begin{align}
    \mse^\proj_\mathrm{iid} =   \frac{JM\sigma^2}{{M} + J\sigma^2}.
\end{align}
Now let us compare the projected variant to the subspace \ac{LMMSE}.
For the case of $\mc=\I_M$, the subspace \ac{LMMSE} estimator boils down to
\begin{align}
    \hat{\vh}_\sub = \frac{1}{1+\sigma^2}\mv \mv^\He\vy,
\end{align}
with its corresponding \ac{MSE} as (cf.~\cref{app:subMSE})
\begin{align}
    \mse^\sub_\mathrm{iid} = \frac{\sigma^2(M\sigma^2 + J)}{(1+\sigma^2)^2} \geq \mse^\proj_\mathrm{iid}. \label{eq:sub_vs_proj}
\end{align}

\cref{fig:nmse_perfect_rayleigh} shows the individual channel estimators' performances based on perfect subspace knowledge for the case of uncorrelated Rayleigh fading.
As one can see, the projected channel estimator outperforms all other estimators across the whole \ac{SNR} range.
Further, we realize that the subspace \ac{LMMSE} converges to the \ac{ML} method from above.

\update{
Additionally, performance guarantees of the proposed projected \ac{LMMSE} compared to the unbiased \ac{CRB} as derived in~\cite{Nayebi2018} can be given.
We formulate the following theorem for the asymptotic region, where $M$ and $N$ tend to infinity.
\begin{theorem}\label{the:perf_proj_crb}
    We consider the limits of $N\to\infty$ and $M/N\to \alpha$ with $\alpha > 0$, while keeping the number of pilot symbols in the order of users $J$.
    Then, the projected \ac{LMMSE} based on perfect knowledge of the subspace spanned by the left singular vectors of the channel matrix $\mh$ is equal to or better than any unbiased estimator in terms of \ac{MSE} for uncorrelated Rayleigh fading.
\end{theorem}
\begin{proof}
We take the limit of the \ac{MSE} of the projected \ac{LMMSE} estimator as
\begin{align}
     \lim_{\substack{N\to\infty\\\frac{M}{N}\to\alpha}}\mse^\proj &= 
    \lim_{\substack{N\to\infty\\\frac{M}{N}\to\alpha}}\frac{J}{M}\sigma^2\sum_{i=1}^M \frac{\rho_i}{\rho_i + \frac{J}{M}\sigma^2} \\
    &= J\sigma^2 \leq (1+\alpha)J\sigma^2 = \mathrm{CRB}_\mathrm{iid,d}, \label{eq:inequ_crb}
\end{align}
where $\mathrm{CRB}_\mathrm{iid,d}$ denotes the deterministic \ac{CRB}'s limit for uncorrelated Rayleigh fading.
The last equality in~\eqref{eq:inequ_crb} was proven in~\cite[Theorem 4]{Zhang2025}.
Furthermore, it is stated in~\cite{Zhang2025} that in this asymptotic limit, the deterministic \ac{CRB} is always lower than the stochastic \ac{CRB}.
\end{proof}
\cref{the:perf_proj_crb} provides some insights on the asymptotic superiority of our proposed projected \ac{LMMSE} strategy based on perfect subspace knowledge compared to any unbiased channel estimator.
Additionally, the \ac{MSE} of the subspace variant as well as the plain \ac{LMMSE} approach infinity for the considered limit, making them irrelevant for this asymptotic regime.
Interestingly, the projected \ac{LMMSE} also approaches $J\sigma^2$ for the case where only $M$ tends to infinity, making the projected \ac{LMMSE} a valid candidate for large antenna systems as the performance no longer depends on $M$.
}%

\subsection{Maximum Likelihood Subspace Estimation}
\label{sec:mle}

After introducing the methods utilizing the additional subspace information provided by $\range(\mv)$ to enhance the \ac{CSI} estimation quality, let us consider the estimation of such a subspace.
As the received data symbols are transmitted over the same channels, we can use these payload symbols to estimate the subspace containing all user channels.

To this end, let us reconsider the \ac{ML} estimate of $\mh$ in~\eqref{eq:ml_opt}.
Instead of directly optimizing on this \ac{ML} formulation as done in~\cite{Nayebi2018},
which generally does not result in the \ac{MMSE}, we only take this log-likelihood formulation as an intermediate step to estimate the subspace $\range(\mv)$.
\update{First, let us consider the payload \ac{ML} objective formulation in~\eqref{eq:orth}.}
We can then reformulate the problem,
\update{neglecting the second term due to phase ambiguity,}
resulting in~\cite{Medles2003} 
\begin{align}
	\max_\mh \;\tr\left(\matp_\mh\hat{\mc}^{(d)}_{\vy\mid\mh}\right), \label{eq:dmle2}
\end{align}
where $\matp_\mh= \mh (\mh^\He\mh)^{-1}\mh^\He = \mv\mv^\He$ and $\hat{\mc}^{(d)}_{\vy\mid\mh} = \frac{1}{N-J}\my_d\my_d^\He$, with $\my_d$ from~\eqref{Eq:AllObservations}.
The maximization in \eqref{eq:dmle2} is solved 
by setting $\matp_\mh$ equal to $\hat{\mv}_d \hat{\mv}_d^\He$
with $\hat{\mv}_d$ holding the $J$ dominant eigenvectors of the receive sample covariance matrix $\hat{\mc}^{(d)}_{\vy\mid\mh}$.
Additionally, it is trivial to see that the first term in~\eqref{eq:ml_opt} is minimized by $\vh_n = \vy(n)$. 
The subspace spanned by the solution $\vh_n = \vy(n)$ is the same as the subspace spanned by the $J$ eigenvectors of the sample covariance matrix $\hat{\mc}^{(p)}_{\vy\mid\mh} = \frac{1}{J}\my_p\my_p^\He$, which ignores the additional phase information contained in the pilot observation.
Thus, 
the overall subspace estimate $\hat{\mv}= [\vv_1,\dots,\vv_J]$ is found by
taking the $J$ dominant eigenvectors $\vv_j$ of the sample covariance matrix defined as
\begin{align}
    \hat{\mc}_{\vy\mid\mh} = \frac{1}{N}\my\my^\He. \label{eq:smplcov_y}
\end{align}
This result has also been used in~\cite{Deng2020}.

To save computational complexity while calculating the sample covariance matrix in~\eqref{eq:smplcov_y}, information from the previous coherence intervals can be utilized.
This can be done by adaptively updating the subspace using efficient tracking algorithms as proposed in, e.g.,~\cite{Utschick2002, Yang1995}.

\begin{figure*}[t]
\centering
\includegraphics[]{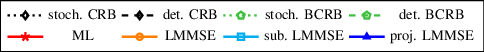}
    \centering
  \subfloat[Perfect Subspace Knowledge]{
	\centering
 \includegraphics[]{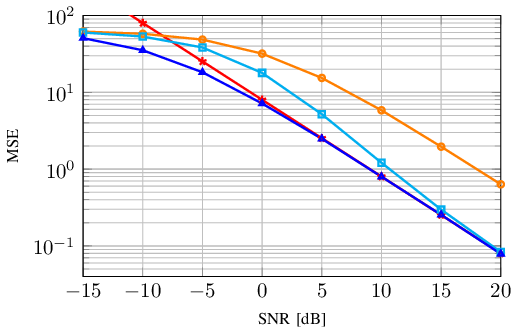}
 \label{fig:nmse_perfect_rayleigh}
     }%
  \subfloat[\update{Subspace Estimation based on $N=200$ Snapshots}]{
	\centering
    \includegraphics[]{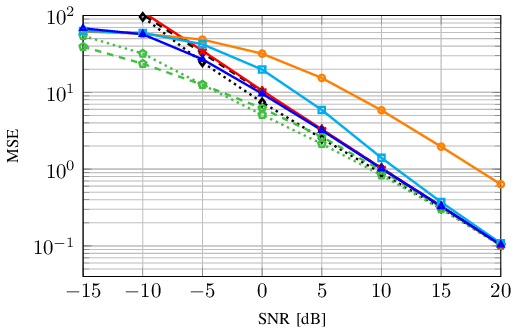}
    \label{fig:nmse_est_rayleigh}
   	}%
	\caption{
	\update{MSE over the SNR for given channel estimations based on perfect statistical knowledge in a $J=8$ user and $M=64$ antennas scenario with uncorrelated Rayleigh fading and Gaussian symbols using (a) perfect subspace knowledge and (b) $N=200$ snapshots for subspace estimation.} 
	}
	\label{fig:nmse_rayleigh}
\end{figure*}

\update{
\subsection{Performance Analysis for Estimated Subspace}
}%

\update{\cref{fig:nmse_est_rayleigh} shows the individual channel estimators' performances based on an estimated subspace using $N=200$ snapshots for uncorrelated Rayleigh fading.
Additionally, the unbiased \acp{CRB} derived in~\cite{Nayebi2018} and their extensions to the \acp{BCRB}, derived in \cref{app:bcrb}, are shown.
For increasing \ac{SNR}, all semi-blind channel estimators considered in this section converge to the \acp{CRB} and \acp{BCRB}.
Furthermore, we see that for low \ac{SNR} values the introduced \ac{MMSE}-based semi-blind channel estimators outperform the different unbiased \acp{CRB}, which are lower bounds for, e.g., the \ac{EM}-based estimator from~\cite{Nayebi2018}.
}

\update{In order to quantify the error introduced by the subspace estimation in comparison to perfect subspace knowledge, let us consider the projected \ac{LMMSE} as introduced in~\cref{sec:proj} in combination with the subspace estimation from~\cref{sec:mle}.
We assume Gaussian distributed symbols for the following derivations.
The introduced error by estimating the subspace can be bounded using~\cref{lem:mse_bound}.
\begin{theorem}\label{lem:mse_bound}
Let the subspace spanned by the left singular vectors of the channel matrix $\mh$ be estimated based on $N$ observations.
The probability that the difference between the projected LMMSE channel estimators utilizing the estimated and true subspace, respectively, is less than $\varepsilon$ is given as
\begin{multline}
    P\left(\E_{\mx,\mn}[\|\hat{\vh}_\proj(\hat{\mv}) - \hat{\vh}_\proj({\mv})\|^2] \leq \varepsilon\right) \\
    \geq 1 - \frac{4 k^2(M-J)\lambda_{\max}(\mw_\proj^\He\mw_\proj)}{J^2N\varepsilon} \sum_{j \leq J}\frac{|s_{\max}|^2+2\sigma^2}{\lambda_j^2}, \label{eq:mse_bound}
\end{multline}
where $\lambda_j$ is the $j$-th eigenvalue of 
$\mh\mh^\He$, 
$s_{\max} = \max_{j\in\{1,\dots,J\}} s_j$ with $\vh=\mv\vs$,
and 
\begin{align}
    k^2 = \sigma^2\left(\sigma^2+\tr(\mc_{\vy\mid\mh})\right).
\end{align}
\end{theorem}
\begin{proof}
    See Appendix~\ref{app:mse_bound}
\end{proof}}

\update{The probability bound in \cref{lem:mse_bound} is very loose, as the used inequalities are not tight.
Nevertheless, this bound can be used to show the superiority of the projected \ac{LMMSE} based on an estimated subspace for a high number of observations compared to the other considered estimators.
\cref{tab:sub_est} gives the bound values for the case of $N=10^4$ observations at \ac{SNR} $=0$ dB for different number of users. 
As one can see, the projected \ac{LMMSE} outperforms the other estimators with high probability.}

\begin{table}
\footnotesize
		\caption{\update{Lower bound on the probability of the projected \ac{LMMSE} outperforming the other estimators for $N=10^4$ observations at \ac{SNR} $=0$dB.}
      }
      \vspace{5pt}
	\centering
    \update{
     \begin{tabular}{c|c|c|c } 
		Number of users $J$& $8$ & $16$ & $32$\\ \hline\hline
		  \rule{0pt}{10pt}$P(\mse^\proj_\mathrm{iid} \leq \mse^\mathrm{sub}_\mathrm{iid})$ & $\geq0.9873$ & $\geq0.9791$ & $\geq0.9642$\\[2pt]  \hline
          \rule{0pt}{10pt}$P(\mse^\proj_\mathrm{iid} \leq \mse^\mathrm{ML}_\mathrm{iid})$ & $\geq0.8453$ & $\geq0.9539$ & $\geq0.9907$\\[2pt]   \hline
          \rule{0pt}{10pt}$P(\mse^\proj_\mathrm{iid} \leq \mse^\mathrm{plain}_\mathrm{iid})$ & $\geq0.9946$ & $\geq0.9923$ & $\geq0.9907$\\[2pt]   
	\end{tabular}}

	\label{tab:sub_est}
\end{table}

\section{Proposed Semi-Blind Channel Estimation - Utilizing Generative Prior}

In practice, the user channels do not follow a Gaussian distribution.
Hence,~\eqref{eq:cme_lmmse} can not be utilized directly and
the \ac{CME} given the pilot observation $\vy_p$ is formulated as
\begin{align}
    \hat{\vh}_\text{CME} = \mathbb{E}\left[\vh\mid\vy_p\right] 
    = \int\vh \frac{p_{\vn}(\vy_p-\vh)p(\vh)}{p(\vy_p)}\mathrm{d}\vh.\label{eq:cme}
\end{align}
As seen in~\eqref{eq:cme}, the \ac{CME} generally can not be computed analytically. 
First, the \ac{CME} needs access to $p(\vh)$, which is generally unavailable in practice. 
Additionally, no closed-form solution exists to the integral in~\eqref{eq:cme}.

In order to reformulate the \ac{CME}, we first use the property that for any arbitrarily distributed random variable $\vh$, we can always find a condition $\vc$ which makes the conditional distribution Gaussian.
Secondly, it has been shown in~\cite{Boeck2024} that for wireless communication channels, this conditional Gaussian distribution preserves the zero-mean property as
\begin{align}
    \vh\mid\vdel \sim \NC (\bm{0},\mc_{\vh\mid\vc}). \label{eq:condModel}
\end{align}
Thus, we can reformulate the \ac{CME} as
\begin{align}
    \mathbb{E}\left[\vh\mid\vy_p\right] &= \mathbb{E}\left[\mathbb{E}\left[\vh\mid\vy_p,\vc\right] \mid\vy_p\right] \\
    &= \int \mathbb{E}\left[\vh\mid\vy_p,\vc\right] p(\vc\mid\vy_p)\mathrm{d}\vc \\
    &= \int \hat{\vh}_\vc(\vy_p) p(\vc\mid\vy_p)\mathrm{d}\vc
\end{align}
where
\begin{align}
    \hat{\vh}_\vc(\vy_p) = \mc_{\vh\mid\vc}\left(\mc_{\vh\mid\vc} + \mc_\vn\right)^{-1}\vy_p, \label{eq:condCME}
\end{align}
denotes the \ac{LMMSE} estimate given $\vc$.
However, finding a suitable condition $\vc$ can be challenging, particularly as the true distribution of $\vh$ is unknown.
To this end, \acp{CGLM} were proposed in~\cite{Koller2022,Fesl2023c,Baur2024}, which approximate the \ac{CME} based on a \ac{GMM}, \ac{MFA}, and \ac{VAE}, respectively.
All three methods learn a model that provides the conditional Gaussian distribution $\vh\mid\vc$ based on a discrete (\ac{GMM}, \ac{MFA}) or continuous (\ac{VAE}) latent variable $\vc$.
This work focuses on the \ac{GMM} and \ac{VAE}, which we adapt to semi-blind channel estimation in the following.

\subsection{GMM-based Semi-blind Channel Estimation}
\label{subsec:GMM}

\begin{algorithm}[t]
\caption{Subspace GMM Channel Estimator}\label{alg:subGMM}
\textbf{Offline Training Phase}
\begin{algorithmic}[1]
\Require Training dataset $\mathcal{H}=\{\vh_t\}_{t=1}^T$
\State Fit the GMM with the EM algorithm, cf.~\cite{Koller2022}
\par\vskip.5\baselineskip\hrule height .4pt\par\vskip.5\baselineskip
\renewcommand{\algorithmicensure}{\textbf{Online Channel Estimation}}
\Ensure
\Require $\my=[\vy(1),\dots,\vy(N)]$, $\matp$, $\sigma^2$
\State $\hat{\mc}_{\vy\mid\mh} \leftarrow \frac{1}{N}\my\my^\He$
\State $\hat{\mv} \leftarrow J$ dominant eigenvectors of $\hat{\mc}_{\vy\mid\mh}$
\State $\my_p=[\vy_{p,1},\dots,\vy_{p,J}] \leftarrow \my^\prime_p\matp^\dagger$
\For {$j = 1,\dots,J$}
\State $\vy^\prime \leftarrow \mv^\He\vy_{p,j}$
\For{$k=1,\dots,K$}
\State $\hat{\vh}^\prime_{k} \leftarrow \mv^\He\mc_k\mv \left(\mv^\He\mc_k\mv + \sigma^2\I_J\right)^{-1}\vy^\prime$ 
\EndFor
\State $\hat{\vh}_j \leftarrow \mv\sum_{k=1}^K p(k\mid\vy^\prime)\hat{\vh}^\prime_{k}$
\EndFor
\State\Return $\hat{\vh}_j, \forall j=1, \dots,J$
\end{algorithmic}
\end{algorithm}

\begin{algorithm}[t]
\caption{Projected GMM Channel Estimator}\label{alg:projGMM}
\textbf{Offline Training Phase}
\begin{algorithmic}[1]
\Require Training dataset $\mathcal{H}=\{\vh_t\}_{t=1}^T$
\State Fit the GMM with the EM algorithm, cf.~\cite{Koller2022}
\par\vskip.5\baselineskip\hrule height .4pt\par\vskip.5\baselineskip
\renewcommand{\algorithmicensure}{\textbf{Online Channel Estimation}}
\Ensure
\Require $\my=[\vy(1),\dots,\vy(N)]$, $\matp$, $\sigma^2$ 
\State $\hat{\mc}_{\vy\mid\mh} \leftarrow \frac{1}{N}\my\my^\He$
\State $\hat{\mv} \leftarrow J$ dominant eigenvectors of $\hat{\mc}_{\vy\mid\mh}$
\State $\my_p=[\vy_{p,1},\dots,\vy_{p,J}] \leftarrow \my^\prime_p\matp^\dagger$
\For {$j = 1,\dots,J$}
\State $\tilde{\vy} \leftarrow \mv\mv^\He\vy_{p,j}$
\For{$k=1,\dots,K$}
\State $\hat{\vh}_{k} \leftarrow \mc_k \left(\mc_k + \mc_{\tilde{\vn}}\right)^{-1}\tilde{\vy}$ 
\EndFor
\State $\hat{\vh}_j \leftarrow \sum_{k=1}^K p(k\mid\tilde{\vy})\hat{\vh}_{k}$
\EndFor
\State\Return $\hat{\vh}_j, \forall j=1, \dots,J$
\end{algorithmic}
\end{algorithm}

Based on the universal approximation property of \acp{GMM}~\cite{Nguyen2020},
the \ac{PDF} of $\vh$ is approximated by
\begin{align}
	f_\vh^{(K)}(\vh) = \sum_{k=1}^{K} p(k) \NC(\vh;\vmu_k,\mc_k), \label{eq:GMM_h}
\end{align}
where $p(k)$, $\vmu_k$, and $\mc_k$ are the mixing coefficients, means, and covariance matrices of the $k$-th \ac{GMM} component, respectively.
As we consider wireless channels, each component's mean is set to $\vmu_k=\bm{0}$, cf.~\cite{Boeck2024}.
The fitting of the components in~\eqref{eq:GMM_h} is accomplished with the well-known \ac{EM} algorithm~\cite{Bishop2006} based on a set $\mathcal{H} = \{\vh_t\}^T_{t=1}$ of $T$ channel
samples as training data.
Based on the formulation in~\eqref{eq:GMM_h} the conditional \ac{PDF} 
is 
\begin{align}
    \vh \mid k \sim \NC(\vh;\bm{0},\mc_k).
\end{align}
Thus, we have a discrete latent variable in the case of a \ac{GMM}, which helps us parameterize the \ac{CME}.
The resulting semi-blind subspace \ac{GMM} can be formulated as
\begin{align}
	\hat{\vh}_\text{sub. GMM} = \mv \hat{\vh}^\prime_\text{GMM} = \mv\sum_{k=1}^K p(k\mid\vy^\prime)\hat{\vh}^\prime_{\text{GMM},k}, \label{eq:subGMM}
\end{align}
with 
\begin{align}
	\hat{\vh}^\prime_{\text{GMM},k} = \;&\mv^\He\mc_k\mv \left(\mv^\He\mc_k\mv + \sigma^2\I_J\right)^{-1}\mv^\He\vy_p, \label{eq:subGMM_k}
\end{align}
and the corresponding responsibilities
\begin{align}
    p(k\mid\vy^\prime) = \frac{p(k)\NC\left(\vy^\prime;\bm{0},\mv^\He\mc_k\mv + \sigma^2\I_J\right)}{\sum_{i=1}^K p(i)\NC\left(\vy^\prime;\bm{0},\mv^\He\mc_i\mv + \sigma^2\I_J\right)}.
\end{align}
The projected \ac{GMM} is
\begin{align}
	\hat{\vh}_\text{proj. GMM} = \sum_{k=1}^K p(k\mid\tilde{\vy})\hat{\vh}_{\text{proj. GMM},k}, \label{eq:projGMM_all}
\end{align}
with
\begin{align}
	\hat{\vh}_{\text{proj. GMM},k} = \mc_k \left(\mc_k + \mc_{\tilde{\vn}}\right)^{-1}\tilde{\vy} \label{eq:projGMM}
\end{align}
and the associated responsibilities
\begin{align}
    p(k\mid\tilde{\vy}) = \frac{p(k)\NC\left(\tilde{\vy};\bm{0},\mc_k + \mc_{\tilde{\vn}}\right)}{\sum_{i=1}^K p(i)\NC\left(\tilde{\vy};\bm{0},\mc_i + \mc_{\tilde{\vn}}\right)}.
\end{align}
The respective estimators are summarized in~\cref{alg:subGMM} and~\cref{alg:projGMM}.

\subsection{VAE-based Semi-blind Channel Estimation}

\begin{figure}
\centering
 \includegraphics[]{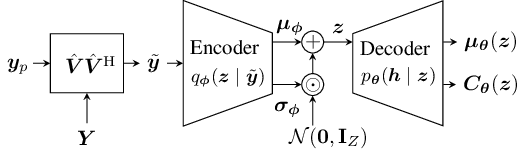}
    \caption{Structure of a semi-blind VAE. The matrix $\hat{\mv}$ contains the $J$ dominant eigenvectors of \eqref{eq:smplcov_y}.
    The encoder and decoder represent DNNs.}
    \label{fig:vae}
\end{figure}

To learn the unknown distribution $p(\vh)$ using a \ac{VAE}, we lower bound the parameterized likelihood $p_\vth(\vh)$ using the \ac{ELBO}.
To formulate the \ac{ELBO}, the variational distributions $q_\vphi (\vz\mid\vy^\prime)$ and $q_\vphi (\vz\mid\tilde{\vy})$ are introduced, which approximate $p(\vz\mid\vy^\prime)$ and $p (\vz\mid\tilde{\vy})$, respectively.
In contrast to the \ac{GMM}, the used subspace $\range(\mv)$ is unknown to the \ac{VAE}'s encoder, making $p(\vz\mid\vy^\prime)$ difficult to learn.
Additionally, the encoder input's dimension would depend on the number of users in the system.
Thus, we propose to approximate both posteriors $p(\vz\mid\vy^\prime)$ and $p (\vz\mid\tilde{\vy})$ with $q_\vphi (\vz\mid\tilde{\vy})$.
A version of the \ac{ELBO} for this case, which is accessible, can be written as~\cite{Kingma2019}
\begin{align}
    \mathcal{L}_{\vth,\vphi} = \mathbb{E}_{q_\vphi} [\log p_\vth (\vh\mid\vz)] - \mathrm{D}_\mathrm{KL}(q_\vphi(\vz\mid\Tilde{\vy})\mid\mid p(\vz)), \label{eq:elbo}
\end{align}
where $\mathbb{E}_{q_\vphi}[\cdot] = \mathbb{E}_{q_\vphi(\vz\mid\tilde{\vy})}[\cdot]$ is the expectation over the variational distribution $q_\vphi (\vz\mid\tilde{\vy})$.
The second term in~\eqref{eq:elbo} is the \ac{KL} divergence
\begin{align}
    \mathrm{D}_\mathrm{KL}(q_\vphi(\vz\mid\Tilde{\vy})\mid\mid p(\vz)) = \mathbb{E}_{q_\vphi}\left[\log\left(\frac{q_\vphi(\vz\mid\Tilde{\vy})}{p(\vz)}\right)\right].
\end{align}
In the \ac{VAE} framework, the \ac{ELBO} is optimized using \acp{DNN} and the reparameterization trick~\cite{Kingma2019}.
In order to do so, the involved distributions are defined as
\begin{align}
    p(\vz) &= \mathcal{N}(\bm{0},\I_Z),\nonumber\\
    p_\vth (\vh\mid\vz) &= \NC (\vmu_\vth(\vz),{\mc}_\vth(\vz)), \\
    q_\vphi (\vz\mid\Tilde{\vy}) &= \mathcal{N}(\vmu_\vphi(\Tilde{\vy}),\diag(\vsig^2_\vphi(\Tilde{\vy}))).\nonumber
\end{align}
The resulting semi-blind \ac{VAE} structure is shown in \cref{fig:vae}.
In the case of a \ac{ULA} or \ac{URA} at the \ac{BS}, the channel covariance matrix is either Toeplitz or block-Toeplitz, respectively. 
As shown in~\cite{Boeck2024}, the conditional covariance matrix at the \ac{VAE}'s output preserves this structure.
Thus, we parameterize the output covariance matrix as 
\begin{align}
    {\mc}_\vth(\vz) = \mq^\He \diag(\vc_\vth(\vz)) \mq,
\end{align}
where $\mq = \mq_M$ or $\mq = \mq^\prime_{N_v} \otimes \mq^\prime_{N_h}$, respectively, where $\mq_M$ is a \ac{DFT} matrix of size $M$ resulting in a circulant approximation, cf.~\cite{Baur2024}, and $\mq^\prime_{N_x}$ contains the first $N_x$ columns of the $2N_x\times2N_x$ \ac{DFT} matrix resulting in a block-Toeplitz parameterization, cf.~\cite{Baur2024a}.
Further, for the (block-)Toeplitz parameterization, we can set $\vmu_\vth(\vz)=\bm{0}$, cf.~\cite{Boeck2024}.

\begin{algorithm}[t]
\caption{Subspace VAE Channel Estimator}\label{alg:subVAE}
\textbf{Offline Training Phase}
\begin{algorithmic}[1]
\Require Training dataset $\mathcal{H}=\{\vh_t\}_{t=1}^T$
\State Fit the VAE by optimizing the ELBO, cf.~\cite{Baur2024}
\par\vskip.5\baselineskip\hrule height .4pt\par\vskip.5\baselineskip
\renewcommand{\algorithmicensure}{\textbf{Online Channel Estimation}}
\Ensure
\Require $\my=[\vy(1),\dots,\vy(N)]$, $\matp$, $\sigma^2$ 
\State $\hat{\mc}_{\vy\mid\mh} \leftarrow \frac{1}{N}\my\my^\He$
\State $\hat{\mv} \leftarrow J$ dominant eigenvectors of $\hat{\mc}_{\vy\mid\mh}$
\State $\my_p=[\vy_{p,1},\dots,\vy_{p,J}] \leftarrow \my^\prime_p\matp^\dagger$
\For {$j = 1,\dots,J$}
\State $\tilde{\vy} \leftarrow \mv\mv^\He\vy_{p,j}$
\State $\vmu_\vth(\vz), \mc_\vth(\vz) \leftarrow \mathrm{VAE}(\tilde{\vy})$
\State $\hat{\vh}_j \leftarrow \mv\mv^\He\mc_\vth(\vz)\mv \left(\mv^\He\mc_\vth(\vz)\mv + \sigma^2\I_J\right)^{-1}$
\State \hspace{\algorithmicindent} $\times(\mv^\He\vy_p - \mv^\He\vmu_\vth(\vz)) - \mv\mv^\He\vmu_\vth(\vz)$
\EndFor
\State\Return $\hat{\vh}_j, \forall j=1, \dots,J$
\end{algorithmic}
\end{algorithm}

\begin{algorithm}[t]
\caption{Projected VAE Channel Estimator}\label{alg:projVAE}
\textbf{Offline Training Phase}
\begin{algorithmic}[1]
\Require Training dataset $\mathcal{H}=\{\vh_t\}_{t=1}^T$
\State Fit the VAE by optimizing the ELBO, cf.~\cite{Baur2024}
\par\vskip.5\baselineskip\hrule height .4pt\par\vskip.5\baselineskip
\renewcommand{\algorithmicensure}{\textbf{Online Channel Estimation}}
\Ensure
\Require $\my=[\vy(1),\dots,\vy(N)]$, $\matp$, $\sigma^2$ 
\State $\hat{\mc}_{\vy\mid\mh} \leftarrow \frac{1}{N}\my\my^\He$
\State $\hat{\mv} \leftarrow J$ dominant eigenvectors of $\hat{\mc}_{\vy\mid\mh}$
\State $\my_p=[\vy_{p,1},\dots,\vy_{p,J}] \leftarrow \my^\prime_p\matp^\dagger$
\For {$j = 1,\dots,J$}
\State $\tilde{\vy} \leftarrow \mv\mv^\He\vy_{p,j}$
\State $\vmu_\vth(\vz), \mc_\vth(\vz) \leftarrow \mathrm{VAE}(\tilde{\vy})$
\State $\hat{\vh}_j \leftarrow \vmu_\vth(\vz) + \mc_\vth(\vz) \left(\mc_\vth(\vz) +\mc_{\tilde{\vn}}\right)^{-1}(\tilde{\vy}-\vmu_\vth(\vz))$
\State \hspace{\algorithmicindent} $\times(\tilde{\vy}-\vmu_\vth(\vz))$
\EndFor
\State\Return $\hat{\vh}_j, \forall j=1, \dots,J$
\end{algorithmic}
\end{algorithm}

After successfully training the \ac{VAE}, 
the output is a local parameterization of $p(\vh)$ as conditionally Gaussian
\begin{align}
    \vh\mid\vz \sim p_\vth(\vh\mid\vz) .
\end{align}
As analyzed in~\cite{Baur2024}, it is a reasonable approximation to set
\begin{align}
    p(\vz\mid\tilde{\vy}) = 
    \begin{cases}
        1 \quad\text{if } \vz = \vmu_\vphi(\tilde{\vy}), \\
        0 \quad\text{otherwise}.
    \end{cases}
\end{align}
Based on this parameterization, we can formulate the semi-blind \ac{VAE}-based estimators as
\begin{align}
	\hat{\vh}_{\text{proj. VAE}} = \;&\vmu_\vth(\vz) + \mc_\vth(\vz) \left(\mc_\vth(\vz) +\mc_{\tilde{\vn}}\right)^{-1}(\tilde{\vy}-\vmu_\vth(\vz)), \label{eq:projVAE}
\end{align}
and
\begin{align}
    	\hat{\vh}_{\text{sub. VAE}} = \;&\mv\mv^\He\mc_\vth(\vz)\mv \left(\mv^\He\mc_\vth(\vz)\mv + \sigma^2\I_J\right)^{-1}\nonumber\\
     &\times(\mv^\He\vy_p - \mv^\He\vmu_\vth(\vz)) - \mv\mv^\He\vmu_\vth(\vz).\label{eq:subVAE}
\end{align}
The respective estimators are summarized in~\cref{alg:subVAE} and~\cref{alg:projVAE}.
For a more detailed introduction into the \ac{VAE} framework and its usage for \ac{CME} parameterization, we refer the reader to~\cite{Baur2024}.

\begin{figure*}[t]
\centering
\includegraphics[]{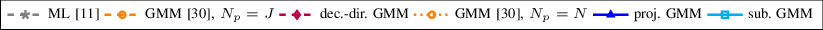}
    \centering
  \subfloat[QPSK]{
	\centering
 \includegraphics[]{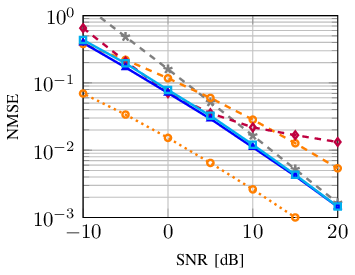}
 \label{fig:nmse_meas_qpsk}
     }%
    \subfloat[16QAM]{
	\centering
 \includegraphics[]{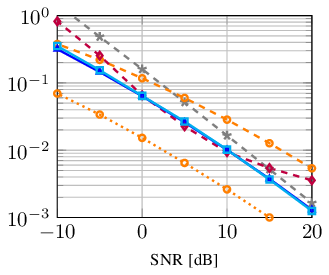}
 \label{fig:nmse_meas_16qam}
     }%
  \subfloat[Gaussian Symbols]{
	\centering
    \includegraphics[]{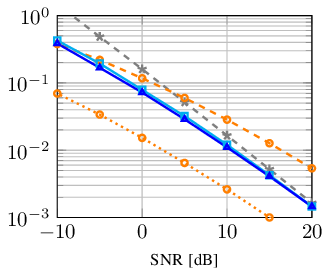}
    \label{fig:nmse_meas_gauss}
   	}%
 \caption{
	NMSE over the \ac{SNR} for given channel estimations in a $J=8$ user and $M=64$ antennas scenario based on $N=200$ observations, including $N_p=J$ pilots (if not stated otherwise) and (a) QPSK, (b) 16QAM, and (c) Gaussian data symbols, using measurement data (Sec.~\ref {sec:meas}).
	}
 \label{fig:nmse_const}
\end{figure*}

\subsection{Complexity Analysis}

The standalone \ac{GMM} estimator proposed by \cite{Koller2022} precomputes the filters used for the individual components, resulting in a complexity of $\mathcal{O}(KM^2)$. 
For the standalone \ac{VAE}, the complexity is given as $\mathcal{O}(DM^2)$~\cite{Baur2024}, where $D$ denotes the number of layers in the \ac{VAE}'s forward pass.
For our semi-blind methods, the subspace calculation requires $\mathcal{O}((N+J)M^2)$.
This results from calculating the sample covariance matrix with $\mathcal{O}(NM^2)$ and taking the eigenvectors of the $J$ largest eigenvalues for the solution of~\eqref{eq:dmle2}.
Using the \ac{PAST} algorithm~\cite{Yang1995}, the computational complexity of the subspace calculation reduces to $\mathcal{O}(JM)$ for every update.
In the case of the subspace \ac{GMM}, the $K$ \ac{LMMSE} estimates can not be precomputed, which results in a complexity of $\mathcal{O}(K(M^2 + JM^2 + J^3))$.
Similarily, the subspace \ac{VAE} exhibits a complexity of $\mathcal{O}(DM^2 + JM^2 + J^3)$.
For the projected versions of the \ac{GMM} and \ac{VAE} the complexity becomes $\mathcal{O}(KM^2 + JM^2)$ and $\mathcal{O}(DM^2 + JM^2)$, respectively.
One should note that the calculation for each of the $K$ components in the \ac{GMM} can be parallelized. 
Similarly, the computations in the \ac{VAE}'s convolutional layers can be parallelized, mitigating the complexity.
\section{Baseline estimators}

The following baseline channel estimators are considered for comparison with our methods.
Based on the found subspace $\range(\mv)$ we can formulate the pilot-based \ac{ML} estimator as $\hat{\vh}_{\text{ML}} = \mv \mv^\He \vy_{p}$,
which is the closed-form solution to~\eqref{eq:ml_est}.
This can be interpreted as the subspace-adjusted version of the conventional least squares (LS) channel estimator given as $\hat{\vh}_\text{LS} = \vy_p$.

Another estimator is based on the sample covariance matrix, which we can compute from the 
training data set ${\mathcal{H}}$ to infer the channels' global statistics as
\begin{align}
 \mc_s = \frac{1}{|{\mathcal{H}}|} \sum_{{\vh}\in {\mathcal{H}}} {\vh}{\vh}^\He. 
    \label{eq:smplcov}
\end{align}
We can use the matrix $\mc_s$ as a statistical prior to parameterize the semi-blind channel estimators outlined in Section \ref{sec:sub} and Section \ref{sec:proj} as 
\begin{align}
	\hat{\vh}_{\text{sub. s-cov}} =\mv\mv^\He\mc_s\mv \left(\mv^\He\mc_s\mv + \sigma^2\I_J\right)^{-1}\mv^\He\vy_{p} ,
\end{align}
and
\begin{align}
	\hat{\vh}_{\text{proj. s-cov}} =\mc_s \left(\mc_s + \mc_{\Tilde{\vn}}\right)^{-1}\matp_\mh\vy_{p}.
\end{align}

\update{
To evaluate the gain achieved by the semi-blind adaptions of the \ac{CGLM} channel estimators, we also compare to the pilot-only \ac{GMM}~\cite{Koller2022} and \ac{VAE}~\cite{Baur2024}.
}

Lastly, we compare our proposed methods to two iterative algorithms optimizing the \ac{ML} formulation in~\eqref{eq:ml_opt}, namely the \ac{EM} from~\cite{Nayebi2018} and a \ac{MP} variant similar to~\cite{Liu2009}, which we run both until convergence or $500$ iterations, whichever comes first.

\section{Numerical Simulations}

\begin{figure*}
\centering
\includegraphics[]{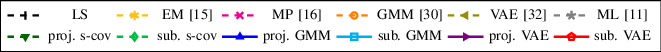}
    \centering
  \subfloat[Spatial Channel Model (Sec.~\ref{sec:3gpp})]{
	\centering
 \includegraphics[]{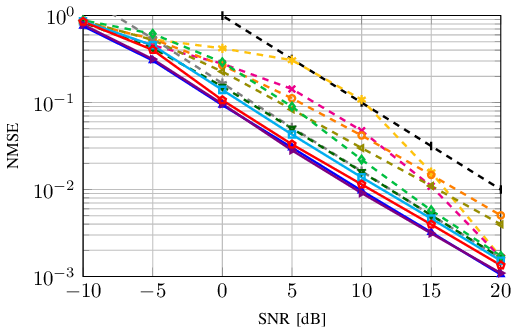}
 \label{fig:nmse_3gpp_SNR}
     }%
  \subfloat[Measurement Data (Sec.~\ref{sec:meas})]{
	\centering
    \includegraphics[]{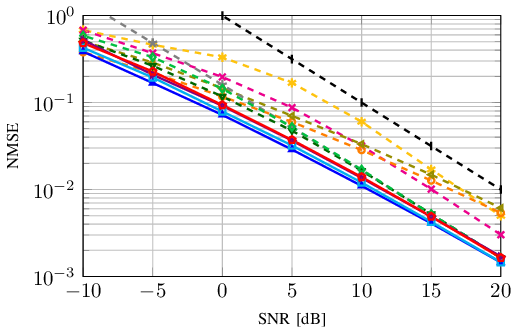}
    \label{fig:nmse_meas_SNR}
   	}%
 \caption{
    NMSE over the \ac{SNR} for given channel estimations in a $J=8$ user and $M=64$ antennas scenario based on $N=200$ observations, including $N_p=J$ pilots and Gaussian symbols, using (a) the spatial channel model and (b) measurement data.
	}
 \label{fig:nmse_SNR}
\end{figure*}

\begin{figure}
\centering
 \includegraphics[]{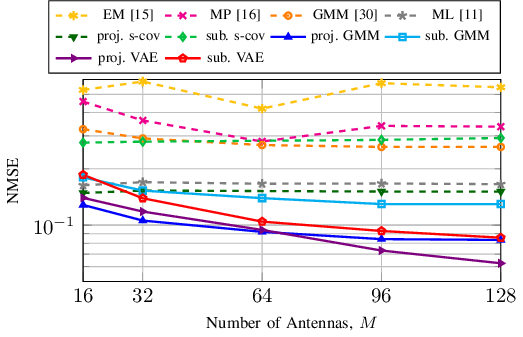}
 \caption{
	NMSE over the number of antennas for given channel estimations in a $J=M/8$ user scenario based on $N=25J$ observations, including $N_p=J$ pilots and Gaussian symbols at \ac{SNR} $=0$ dB, using the spatial channel model.
	}
 \label{fig:nmse_antenna}
\end{figure}

To evaluate our proposed methods, we use channel realizations, which are normalized with $\mathbb{E}\left[\|\vh\|^2\right]=M$.
Thus, we can define the \ac{SNR} as $\sigma^{-2}$.
Further, the \ac{NMSE} defined as  
\begin{align}
    \text{NMSE} = \frac{1}{ML}\sum_{\ell=1}^L \|\vh_\ell - \hat{\vh}_\ell\|^2,
\end{align}
is used to characterize the estimators' performances based on
$L = 10^3$ unseen channel samples stemming from the channel models detailed in \cref{sec:3gpp,sec:meas}.
We use $\mathcal{H} = \{\vh_t\}^T_{t=1}$ with $T=1.5\cdot10^5$ training samples from the respective channel model to train the \ac{GMM} and \ac{VAE}, where we set the number of components to $K=64$ and the latent dimension to $Z=32$, respectively.
Further, in the case of the \ac{VAE} we allow non-zero values for $\vmu_\vth(\vz)$, as we use the circulant approximation for the spatial channel model and the block-Toeplitz property is not perfectly fulfilled for the measurement data, due to hardware imperfections.
\update{The \ac{VAE}'s remaining implementation details are the same as provided in~\cite{Baur2024}.}
The assumption of spatial uncorrelated channels, as used in \eqref{eq:noise_tilde}, only holds for the spatial channel model of \cref{sec:3gpp}.
For the case of the measurement campaign described in \cref{sec:meas}, we approximate the noise covariance matrix in~\eqref{eq:cov_y} 
as
\begin{align}
	\mc_{\tilde{\vn}} \approx \sigma^2\frac{J}{M}\I_M. \label{eq:proj_noise_approx}
\end{align}
\update{Using the $T$ training samples to empirically approximate the projected noise covariance matrix~\eqref{eq:cov_y} did not result in a different performance compared to using the approximation~\eqref{eq:proj_noise_approx}.}
The ``s-cov'' variants (``sub. s-cov'' and ``proj. s-cov'') utilize the same training samples. 
For most of the simulations, the number of \ac{BS} antennas is set to $M=64$, cf.~\cref{sec:3gpp,sec:meas}, serving $J=8=M/8$ number of users, a representative operating point~\cite[Chap. 1.3.3]{Bjoernson2017}.
Further, if not stated otherwise, the number of snapshots is set to $N=200$, corresponding to a scenario that allows high channel dispersion and mobility, e.g., up to $135$ kph, c.f.~\cite[Chap. 2.1]{Bjoernson2017}.

\update{
In~\cref{fig:nmse_const}, the channel estimation performances are compared for different symbol constellations using the proposed \ac{GMM} based semi-blind channel estimation variants.
Additionally, we compare to the purely pilot-based \ac{GMM} channel estimator with $N_p=J$ and $N_p=N$ number of pilots, where the latter assumes all sent symbols to be known by the receiver.
We utilize Gaussian symbols with $x_j(n) \sim \NC (0, P_j=1/J)$ such that $\sum_{j=1}^J P_j =1$ as well as QPSK and $16$QAM.
In the latter two cases, also a decision directed \ac{GMM} is evaluated, which decodes the sent symbols based on a pilot-based \ac{GMM} channel estimate and uses the decoded symbols as additional pilots in a second stage.
For the decoding, we utilize \ac{LMMSE} equalization with a subsequent mapping to the closed constellation point.}

\update{
Firstly, we see no qualitative difference between the different symbol constellations.
}%
Using a continuous symbol constellation has a negligible effect on the simulation results, as also previously observed in~\cite{Nayebi2018} \update{and, hence, for the rest of this work, we utilize Gaussian symbols.}
\update{
Further, we see in~\cref{fig:nmse_const} an improvement from the pilot-only \ac{GMM} to our proposed semi-blind variants, showcasing again the advantage of utilizing the payload symbols in channel estimation.
The full pilot \ac{GMM} using $N_p=N$ pilot symbols achieves the \ac{SNR} improvement as shown in~\cref{sec:sys}.
The decision directed \ac{GMM} based estimator saturates for high \ac{SNR} values, making it an inferior semi-blind strategy.
}%

\begin{figure*}
\centering
\includegraphics[]{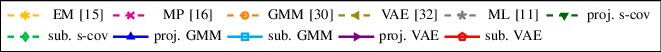}
    \centering
  \subfloat[Spatial Channel Model (Sec.~\ref{sec:3gpp})]{
	\centering
 \includegraphics[]{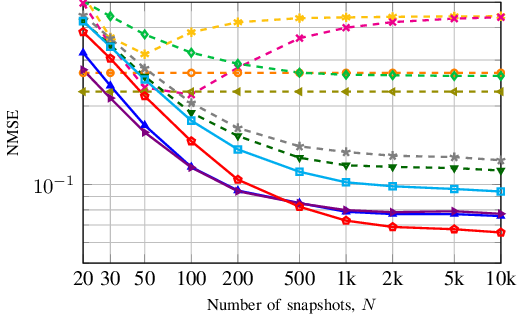}
 \label{fig:nmse_3gpp_snaps}
     }%
  \subfloat[Measurement Data (Sec.~\ref{sec:meas})]{
	\centering
    \includegraphics[]{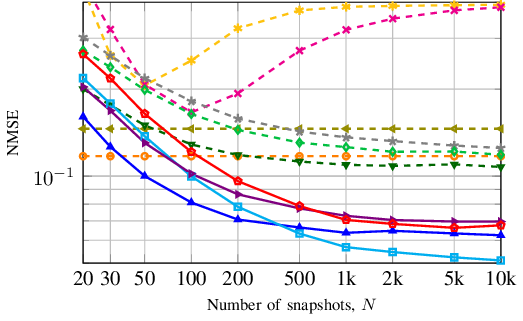}
    \label{fig:nmse_meas_snaps}
   	}%
\caption{
	NMSE over the number of observations for given channel estimations in a $J=8$ user and $M=64$ antennas scenario based on $N$ observations, including $N_p=J$ pilots and Gaussian symbols at \ac{SNR} $=0$ dB, using (a) the spatial channel model and (b) measurement data.
	}
 \label{fig:nmse_snaps}
\end{figure*}

\cref{fig:nmse_3gpp_SNR} and \cref{fig:nmse_meas_SNR} show the different channel estimators' performances with respect to the \ac{SNR} for the spatial channel model (cf. \cref{sec:3gpp}) and measurement data (cf. \cref{sec:meas}), respectively. 
One can see that the semi-blind methods utilizing the \acp{CGLM} perform the best across the whole \ac{SNR}.
The projected variants of the \acp{CGLM} slightly outperform their subspace counterpart for most \ac{SNR} values, which follows the derivations in~\cref{sec:perf}.
Interestingly, the order of the semi-blind \ac{GMM} and semi-blind \ac{VAE} depends on the utilized channel model.
In \cref{fig:nmse_3gpp_SNR}, the projected \ac{GMM} and \ac{VAE} show both the best overall result, whereas in \cref{fig:nmse_meas_SNR}, the projected \ac{GMM} outperforms the \ac{VAE}-based estimator. 
Additionally, we see that in \cref{fig:nmse_3gpp_SNR}, the subspace \ac{VAE} outperforms the subspace \ac{GMM}, and in \cref{fig:nmse_meas_SNR}, the results are vice versa.
This ordering follows the ordering of the standalone version, where the plain \ac{GMM} is better than the plain \ac{VAE} in the case of the measurement data and worse for the spatial channel model.
\update{This comes from the fact that the Toeplitz assumption at the \ac{VAE} may not be fulfilled due to impairments in the real measurement data}.
For high \ac{SNR} values, all semi-blind variants approach each other except the \ac{EM} and \ac{MP} methods.
In \cref{fig:nmse_3gpp_SNR}, the \ac{EM} and \ac{MP} drastically improve from $15$dB to $20$dB, showing similar performance at $20$dB as the other semi-blind methods,
whereas, in \cref{fig:nmse_meas_SNR}, they show inferior results also for high \ac{SNR}.
The \ac{CGLM}-based approaches keep a slight advantage even for high \ac{SNR} values, which can be attributed to the fact that prior information is beneficial even for high \ac{SNR}.
A notable observation is that, in the mid-\ac{SNR} range, the semi-blind \ac{CGLM} variants outperform all related estimators by roughly $3$ dB.

\update{In~\cref{fig:nmse_antenna} we evaluate the respective estimators across different numbers of antennas at the \ac{BS}.
In order to make the comparison fair, we keep the number of users at a constant ratio of $J=M/8$ and set the number of snapshots to $N=25J$.
With this, we only see an improvement for higher numbers of antennas for the \acp{CGLM}, while the other methods roughly keep their performance.
Due to the constant parameter ratio, the \ac{VAE} variants' drastic performance gains for a higher number of antennas solely come from the \ac{VAE}'s strong capability of modeling high-dimensional data.
}%

For our proposed strategies, the accuracy of the estimated subspace $\range(\hat{\mv})$ influences the performance and, hence, the \ac{NMSE} depends on the number of snapshots $N$ as shown in \cref{fig:nmse_snaps}.
We see that for an increasing number of snapshots, the \ac{NMSE} of our proposed methods decreases.
In the case of the spatial channel model (\cref{fig:nmse_3gpp_snaps}), the projected \ac{GMM} and \ac{VAE} perform best for low numbers of snapshots, where the standalone \ac{VAE} surpasses all other methods for $N=20$. 
Additionally, we observe in \cref{fig:nmse_3gpp_snaps} that for high $N$, the subspace \ac{VAE} becomes the best of all considered methods.
In \cref{fig:nmse_meas_snaps}, we observe again that for the measurement data, the semi-blind \ac{GMM} variants perform the best, where for high $N$ the subspace \ac{GMM} and low $N$ the projected \ac{GMM} outperforms all other methods.
Again, for less than $30$ snapshots, the semi-blind methods are outperformed by the standalone \ac{GMM} and \ac{VAE} due to inaccuracies in estimating the subspace with a low number of payload data symbols.
For both utilized channel models, the superior \ac{CGLM} subspace variant surpasses the projected variant for high numbers of snapshots, converging to a lower error level.
Thus, in practice, where, in general, uncorrelated Rayleigh fading is not the case, there are cases where the subspace \ac{CGLM} outperforms its projected counterpart.

A critical decrease in performance can be observed for the \ac{EM} and \ac{MP} algorithms.
Here, the \ac{NMSE} increases after a certain point when increasing the number of snapshots. 
Even though the minimum appears at different $N$, the overall behavior exhibits similarities. 
This is because both methods optimize the joint \ac{ML} formulation in~\eqref{eq:ml_opt}, where the optimization of the second term becomes dominant for a high number of snapshots.
Hence, the pilot observation's impact, which is relevant to estimating the channel phase, vanishes.

\begin{figure*}
\centering
\includegraphics[]{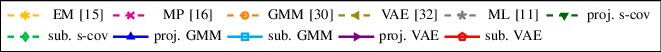}
    \centering
  \subfloat[Spatial Channel Model (Sec.~\ref{sec:3gpp})]{
	\centering
 \includegraphics[]{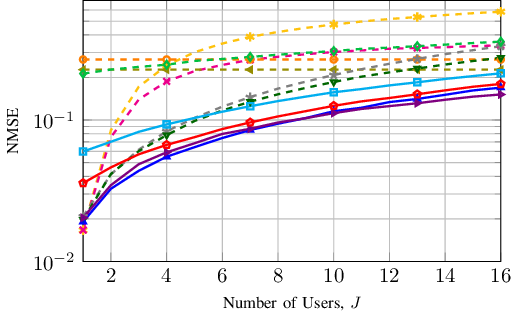}
 \label{fig:nmse_3gpp_user}
     }%
  \subfloat[Measurement Data (Sec.~\ref{sec:meas})]{
	\centering
    \includegraphics[]{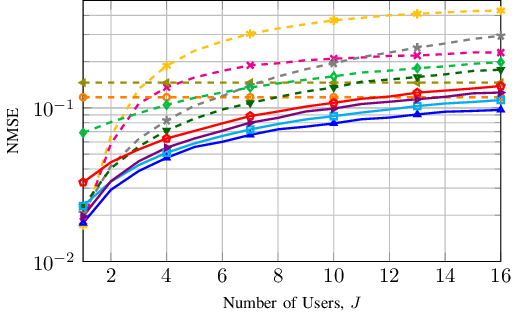}
    \label{fig:nmse_meas_user}
   	}%
 \caption{
	NMSE over the number of users for given channel estimations in a $M=64$ antennas scenario based on $N=200$ observations, including $N_p=J$ pilots and Gaussian symbols at \ac{SNR} $=0$ dB, using (a) the spatial channel model and (b) measurement data.
	}
 \label{fig:nmse_user}
\end{figure*}

The dimension of the subspace $\range(\hat{\mv})$ directly influences the proposed methods' estimation qualities as shown in \cref{fig:nmse_3gpp_user} and \cref{fig:nmse_meas_user}.
For example, in the extreme case where the number of users in the system is equal to the number of \ac{BS} antennas ($J=M$), the solution to~\eqref{eq:dmle2} becomes $\hat{\mv}\hat{\mv}^\He = \I$ and, hence, as the number of users in the system increases all semi-blind estimators approach their respective purely pilot based version.
We restrict our simulations within the interval of $J \leq M/4=16$, which is said to be the preferred operating regime in massive MIMO~\cite[Chap. 1.3.3]{Bjoernson2017}, and set the number of snapshots to $N=200$.
In the case of a single user, all semi-blind variants exhibit similar performance, except for the subspace sample covariance estimator, the subspace \ac{VAE}, and in the case of the spatial channel model (\cref{fig:nmse_3gpp_user}), the subspace \ac{GMM}.
For all other considered numbers of users, the proposed projected \ac{CGLM} methods outperform all other channel estimators.
Additionally, for the spatial channel model in \cref{fig:nmse_3gpp_user}, the subspace \ac{GMM} also shows inferior results to the other \ac{CGLM}-based methods for all numbers of users.

Overall, we can conclude that the proposed semi-blind \acp{CGLM} show superior channel estimation performance across all different setups.
Depending on the used channel model, either the semi-blind \acp{GMM} or the semi-blind \acp{VAE} result in slightly better \ac{NMSE}, where only in the case of the spatial channel model, the subspace \ac{GMM} shows slightly worse performance compared to the other proposed methods.
Moreover, the projected \acp{CGLM} outperform their respective subspace counterpart for most simulated operating points, showing the proposed projection method's superiority.

\section{Conclusion}

This work presented a novel semi-blind channel estimation technique based on the class of \acp{CGLM}.
To this end, two methods that incorporate subspace knowledge into the well-known \ac{LMMSE} estimator are discussed.
Both methods exploit the estimated subspace derived from the dominant eigenvectors of sample covariance matrices constructed using the received symbols.
A theoretical analysis of the methods showed the proposed projection-based estimator's superior estimation quality for uncorrelated Rayleigh fading channels.
Secondly, we showed how two examples from the class of \acp{CGLM}, i.e., the \ac{GMM} and \ac{VAE}, can be used to parameterize these estimators. 
Extensive simulations based on real-world measurement and spatial channel model data demonstrated the proposed methods' superior estimation performances compared to standard semi-blind channel estimators.

\appendix

\subsection{MSE of Projected LMMSE}
\label{app:projMSE}

For any linear estimator $\hat{\vh}=\mw\vy_p$, the \ac{MSE} is given as
\begin{align}
    \mse &= \E\left[\|\vh-\hat{\vh}\|^2\right] = \E\left[\tr\left(\left(\vh-\hat{\vh}\right)\left(\vh^\He-\hat{\vh}^\He\right)\right)\right]\\
    &=\E\left[\tr(\vh\vh^\He) - 2\tr(\vh\vy^\He\mw^\He) + \tr(\mw\vy\vy^\He\mw^\He)\right]. \label{eq:mse_terms}
\end{align}
For the case of the projected \ac{LMMSE} the second term in~\eqref{eq:mse_terms} can be rewritten as
\begin{align}
    &\E\left[\tr\left(\vh\Tilde{\vy}^\He\mw^\He\right)\right] \\
    &\quad\quad= \E\left[\tr\left(\vh\vh^\He\mw^\He\right)\right]\\
    &\quad\quad= \E\left[\tr\left(\vh\vh^\He\left(\mc+\sigma^2\frac{J}{M}\I_M\right)^{-1}\mc\right)\right] \\
    &\quad\quad= \tr\left(\mc\left(\mc+\sigma^2\frac{J}{M}\I_M\right)^{-1}\mc\right).
\end{align}
Similarly for the third term in~\eqref{eq:mse_terms} we have
\begin{align}
    &\E\left[\tr\left(\mw\Tilde{\vy}\Tilde{\vy}^\He\mw^\He\right)\right] \\
    &\quad\quad= \E\left[\tr\left(\mw\left(\vh\vh^\He+\Tilde{\vn}\Tilde{\vn}^\He\right)\mw^\He\right)\right] \\
    &\quad\quad= \E\bigg[\tr\bigg(\mc\left(\mc+\sigma^2\frac{J}{M}\I_M\right)^{-1}\left(\vh\vh^\He+\Tilde{\vn}\Tilde{\vn}^\He\right)\nonumber\\
    &\quad\quad\quad\times\left(\mc+\sigma^2\frac{J}{M}\I_M\right)^{-1}\mc\bigg)\bigg]\\
    &\quad\quad= \tr\left(\mc\left(\mc+\sigma^2\frac{J}{M}\I_M\right)^{-1}\mc\right),
\end{align}
where we assume that $\E\left[\Tilde{\vn}\Tilde{\vn}^\He\right]=\sigma^2\frac{J}{M}\I_M$.
From this, the overall \ac{MSE} in~\eqref{eq:mse_proj2} follows directly.

\subsection{MSE of Subspace LMMSE for Rayleigh Fading}
\label{app:subMSE}

In the case of uncorrelated Rayleigh fading, the subspace \ac{LMMSE} filter is given as
\begin{align}
    \mw_\sub = \frac{1}{1+\sigma^2}\mv\mv^\He.
\end{align}
Using this filter the second term in~\eqref{eq:mse_terms} can be rewritten as
\begin{align}
    &\E\left[\tr\left(\vh\vy^\He\mw_\sub^\He\right)\right] \\
    &\quad\quad= \E_\vh\left[\E\left[\tr\left(\vh\vh^\He\mw_\sub^\He\right)\mid\vh\right]\right]\\
    &\quad\quad= \frac{1}{1+\sigma^2}\E_\vh\left[\E\left[\tr\left(\vh\vh^\He\mv\mv^\He\right)\mid\vh\right]\right] \\
    &\quad\quad= \frac{1}{1+\sigma^2}\E_\vh\left[\tr\left(\vh\vh^\He\right)\right] \\
    &\quad\quad= \frac{1}{1+\sigma^2} M.
\end{align}
Similarly for the third term in~\eqref{eq:mse_terms} we have
\begin{align}
    &\E\left[\tr\left(\mw_\sub\vy\vy^\He\mw_\sub^\He\right)\right] \\
    &\quad\quad= \E_\vh\left[\E\left[\tr\left(\mw_\sub\vh\vh^\He\mw_\sub^\He\right)\mid\vh\right]\right] \nonumber\\
    &\quad\quad\quad+ \E_\vh\left[\E\left[\tr\left(\mw_\sub\vn\vn^\He\mw_\sub^\He\right)\mid\vh\right]\right]\\
    &\quad\quad= \frac{1}{(1+\sigma^2)^2}\Big[\E_\vh\left[\E\left[\tr\left(\mv\mv^\He\vh\vh^\He\mv\mv^\He\right)\mid\vh\right]\right] \nonumber\\
    &\quad\quad\quad+ \E_\vh\left[\E\left[\tr\left(\mv\mv^\He\vn\vn^\He\mv\mv^\He\right)\mid\vh\right]\right]\Big]\\
    &\quad\quad= \frac{1}{(1+\sigma^2)^2}\left[\E_\vh\left[\tr\left(\vh\vh^\He\right)\right]+\E_\vh\left[\sigma^2\tr\left(\mv\mv^\He\right)\right] \right]\\
    &\quad\quad= \frac{1}{(1+\sigma^2)^2} (M + J\sigma^2).
\end{align}
The overall subspace variant's \ac{MSE} for $\mc=\I_M$ is then
\begin{align}
    \mse_\mathrm{iid}^\sub &= M - 2\frac{1}{1+\sigma^2}M + \frac{1}{(1+\sigma^2)^2}(M+J\sigma^2) \\
    &= \frac{\sigma^2(M\sigma^2 + J)}{(1+\sigma^2)^2}.
\end{align}

\update{
\subsection{Derivation of Bayesian Cramer-Rao Bounds}
\label{app:bcrb}
The \ac{BCRB}, which is the lower bound on the \ac{MSE} of any estimator $\hat{\vh}(\my)$, is given as~\cite{VanTrees2001}
\begin{align}
    \mathbb{E}\left[\|\hat{\vh}(\my) - \vh\|^2\right] \geq
    \tr\left(\left[\left(\mathbb{E}_{{\vh}}[\mathcal{I}({{\vh}})] + \bm{J}_p\right)^{-1}\right]_{:M,:M}\right),\label{eq:bcrb_general}
\end{align}
where $\mathcal{I}({{\vh}})$ is the Fisher information matrix corresponding to the inverse of the unbiased \ac{CRB},
and
\begin{align}
    \bm{J}_p 
    &= \mathbb{E}_\vh\left[\frac{\partial\log p(\vh)}{\partial \vh^\He}  \frac{\partial \log p(\vh)}{\partial \vh} \right] \\
    &= \mathrm{blckdiag}({\mc}_1,\dots,{\mc}_J)^{-1}
    ,
\end{align}
We denote with $[\cdot]_{:M,:M}$ in~\eqref{eq:bcrb_general} the part of the overall \ac{BCRB} matrix corresponding to the upper left block of size $M\times M$, where we assume without loss of generality to index the users such that the user of interest is the first one.
}%

\update{
For the case of $\mc_j=\I_M$ for all $j\in\{1,\dots,J\}$ the prior information $\bm{J}_p$ is given as
\begin{align}
    \bm{J}_p = \I_{JM}.
\end{align}
For the deterministic \ac{BCRB}, we use the formulation derived in~\cite{Zhang2025}, which results in
\begin{align}
    &\mathbf{BCRB}_\mathrm{iid,d} \nonumber\\
    &\quad=\sigma^2\left((\mx^*\mx^\T) \otimes \matp_\mh^\perp + (\matp^*\matp^\T)\otimes \matp_\mh + \sigma^2\I_{JM}\right)^{-1},
\end{align}
with $\matp_\mh^\perp = \I_M - \matp_\mh$.
}

\update{
Now using the Fisher information matrix for the stochastic \ac{CRB} as provided in~\cite{Nayebi2018}, we introduce $\Bar{\vh} = [\Re\{\vh_1^\T\}, \Im\{\vh_1^\T\},\dots,\Re\{\vh_J^\T\}, \Im\{\vh_J^\T\}]^\T$, which gives us the prior information as
\begin{align}
    \bm{J}_p &=\mathbb{E}_{\Bar{\vh}}\left[\frac{\partial\log p(\Bar{\vh})}{\partial \Bar{\vh}^\T}  \frac{\partial\log p(\Bar{\vh})}{\partial \Bar{\vh}} \right] \\
    &=2\, \mathrm{blckdiag}(\Bar{\mc}_1,\dots,\Bar{\mc}_J)^{-1},
\end{align}
with
\begin{align}
    \Bar{\mc}_j =
    \begin{bmatrix}
        \Re\{\mc_j\} & -\Im\{\mc_j\} \\
        \Im\{\mc_j\} & \Re\{\mc_j\}
    \end{bmatrix}.
\end{align}
Thus, the \ac{BCRB} for uncorrelated Rayleigh fading is given as
\begin{align}
    \mathbf{BCRB}_\mathrm{iid,s} = \left(\left(\frac{2N_p}{J\sigma^2} + 2\right)\I_{2JM} + \mr\right)^{-1},
\end{align}
where $\mr$ denotes the part depending on the payload symbols defined in~\cite{Nayebi2018}.
}

\update{
\subsection{Proof of \cref{lem:mse_bound}}\label{app:mse_bound}
To this end, let us consider the \ac{MSE} for a fixed channel matrix $\mh$, where the expectation is only taken over the noise and Gaussian symbol realizations as
\begin{align}
    &\E_{\mx,\mn} [\|\hat{\vh}_\proj(\hat{\mv}) - \hat{\vh}_\proj({\mv})\|^2] \nonumber\\
    &\quad= \E_{\mx,\mn} [\|\mw_\proj (\hat{\mv}\hat{\mv}^\He - \mv\mv^\He)\vy\|^2] \\
    &\quad\leq\lambda_{\max}(\mw_\proj^\He\mw_\proj)\E_{\mx}\left[(\tr(\ma) + \E_{\mn}[\tr(\mb)])\right]
\end{align}
with
\begin{align}
    \tr(\ma) &= \tr((\hat{\mv}\hat{\mv}^\He - \mv\mv^\He)\vh\vh^\He(\hat{\mv}\hat{\mv}^\He - \mv\mv^\He)) \\
    &= \tr((\I - \hat{\mv}\hat{\mv}^\He)\vh\vh^\He(\I - \hat{\mv}\hat{\mv}^\He)) \\
    &= \|(\I - \hat{\mv}\hat{\mv}^\He)\vh\|^2 \\
    &= \|(\I - \hat{\mv}\hat{\mv}^\He)\mv\vs\|^2 \\
    &= \sum_{i\leq J}\sum_{j>J}|\hat{\vv}_j^\He\vv_is_i|^2 \\
    &\leq |s_{\max}|^2\sum_{j\leq J}\sum_{i>J}|\hat{\vv}_j^\He\vv_i|^2 \\
    &= |s_{\max}|^2\sum_{i\leq J}\sum_{j>J}|\hat{\vv}_j^\He\vv_i|^2
\end{align}
where $\vh=\mv\vs$ with $\vs=[s_1,\dots,s_J]^\T$ is the channel of interest's decomposition in terms of $\mv$, and $s_{\max}=\max_i s_i$.
Furthermore, we have
\begin{align}
    \E_\mn[\tr(\mb)] &= \E_\mn[\tr((\hat{\mv}\hat{\mv}^\He - \mv\mv^\He)\vn\vn^\He(\hat{\mv}\hat{\mv}^\He - \mv\mv^\He))] \\
    &= \sigma^2 (\tr(\hat{\mv}\hat{\mv}^\He + \mv\mv^\He) - 2\tr(\hat{\mv}\hat{\mv}^\He\mv\mv^\He)) \\
    &= 2\sigma^2 \left[J - \sum_{i,j \leq J}|\hat{\vv}_j^\He\vv_i|^2\right] \\
    &= 2\sigma^2 \sum_{j\leq J}\sum_{i> J}|\hat{\vv}_j^\He\vv_i|^2
\end{align}
Thus,
\begin{align}
    &\E_{\mn} [\|\hat{\vh}_\proj(\hat{\mv}) - \hat{\vh}_\proj({\mv})\|^2] \nonumber\\
    &\quad\leq\lambda_{\max}(\mw_\proj^\He\mw_\proj)\left(|s_{\max}|^2+2\sigma^2\right)\sum_{j\leq J}\sum_{i> J}|\hat{\vv}_j^\He\vv_i|^2.
\end{align}
Now, let us consider the following corollary from~\cite[Corollary 4.1]{Loukas2017} about the inner product of eigenvectors of sample and true covariance matrices.
\begin{corollary}
Let us denote with $\lambda^\prime_j$ and $\vv_j$ the $i$-th eigenvalue and eigenvector of the true covariance matrix, respectively.
Then, for any weights $w_{ji}$, which are non-zero when $\lambda^\prime_i\neq\lambda^\prime_j$ and $\mathrm{sgn}(\lambda^\prime_i - \lambda^\prime_j)2\lambda^\prime_i > \mathrm{sgn}(\lambda^\prime_i - \lambda^\prime_j)(\lambda^\prime_i + \lambda^\prime_j)$,
and real $\varepsilon>0$
    \begin{align}
    &P\left(\sum_{j\leq J}\sum_{i > J}w_{ji}|\hat{\vv}_j^\He\vv_i|^2 \leq \varepsilon\right) \nonumber\\
    &\quad\geq 1 - \sum_{j \leq J}\sum_{i > J}\frac{4 k^2_iw_{ji}}{N\varepsilon(\lambda^\prime_j - \lambda^\prime_i)^2},
\end{align}
where $k_i^2=\mathbb{E}[\|\vy\vy^\He\vv_i\|^2] - \lambda_i^{\prime,2}$.
\end{corollary}
In our given problem, the true covariance matrix is given as $\mc_{\vy\mid\mh} = \mh\mh^\He/J + \sigma^2\I_M$.
It is then trivial to see that $\lambda^\prime_i = \sigma^2, \forall i > J$.
Furthermore it is shown that~\cite[Corollary 4.3]{Loukas2017}
\begin{align}
    k_i^2 = \lambda^\prime_i\left(\lambda^\prime_i+\tr(C_{\vy\mid\mh})\right). 
\end{align}
Hence, we can conclude
\begin{align}
    &P\left(\E_{\mx,\mn}[\|\hat{\vh}_\proj(\hat{\mv}) - \hat{\vh}_\proj({\mv})\|^2]\leq \varepsilon\right) \nonumber\\
    &\geq 1 - \frac{4 k^2(M-J)\lambda_{\max}(\mw_\proj^\He\mw_\proj)}{J^2N\varepsilon} \sum_{j \leq J}\frac{|s_{\max}|^2+2\sigma^2}{\lambda_j^2},
\end{align}
with $\lambda_j = (\lambda^\prime_j - \sigma^2)/J$ for all $j \in \{1,\dots,J\}$ being the eigenvalues of $\mh\mh^\He$ and
\begin{align}
    k^2 = \sigma^2\left(\sigma^2+\tr(C_{\vy\mid\mh})\right).
\end{align}
\qed
}%

\balance
\bibliographystyle{IEEEtran}
  
\bibliography{IEEEabrv,mybib}

\end{document}